\documentclass[journal]{IEEEtran}
\usepackage{graphicx}
\usepackage{booktabs}
\usepackage{latexsym,bm,amsmath,amssymb}
\usepackage{epsfig}
\usepackage{graphicx}
\usepackage{subfigure}
\usepackage{multirow}
\usepackage{ragged2e}
\usepackage[linesnumbered,ruled,lined]{algorithm2e}
\usepackage{color}
\usepackage{mathrsfs}
\usepackage{array}
\usepackage{enumerate}
\usepackage{xcolor}
\usepackage[misc]{ifsym}
\usepackage{cite}
\usepackage{threeparttable}
\usepackage{setspace}
\usepackage[colorlinks, linkcolor=red]{hyperref}
\usepackage[version=4]{mhchem}
\SetKw{Continue}{continue}
\SetKw{Boolean}{boolean}
\SetKw{Return}{return}
\newtheorem{theorem}{Theorem}

\newtheorem{proof}{Proof}
\begin{document}
%
\title{A UAV-Enabled Time-Sensitive Data Collection Scheme for Grassland Monitoring Edge Networks}
%
%
%

\author{\IEEEauthorblockN{Dongbin~Jiao\IEEEauthorrefmark{1}\IEEEauthorrefmark{2},~\IEEEmembership{Member,~IEEE,}
Zihao~Wang\IEEEauthorrefmark{1}, Wen~Fan\IEEEauthorrefmark{1}, Weibo~Yang\IEEEauthorrefmark{5}, Peng~Yang\IEEEauthorrefmark{3}\IEEEauthorrefmark{4}, ~\IEEEmembership{Senior Member,~IEEE}, 
Zhanhuan Shang\IEEEauthorrefmark{6}, and Shi Yan\IEEEauthorrefmark{1},~\IEEEmembership{Member,~IEEE}}
\thanks{\IEEEauthorblockA{\IEEEauthorrefmark{1}School of Information Science and Engineering, Lanzhou University, Lanzhou, 730000, P. R. China (e-mail: \{jiaodb, wangzihao2020, fanw20, yanshi\}@lzu.edu.cn).}}
\thanks{\IEEEauthorblockA{\IEEEauthorrefmark{2}Key Laboratory of Tourism Information Fusion Processing and Data Ownership Protection, Ministry of Culture and Tourism, Lanzhou University, Lanzhou, 730000, P. R. China.}}
\thanks{\IEEEauthorblockA{\IEEEauthorrefmark{3}Guangdong Provincial Key Laboratory of Brain-inspired Intelligent Computation, Department of Computer Science and Engineering, Southern University of Science and Technology, Shenzhen, 518055, P. R. China (e-mail: yangp@sustech.edu.cn).}}
\thanks{\IEEEauthorblockA{\IEEEauthorrefmark{4}Department of Statistics and Data Science, Southern University of Science and Technology, Shenzhen 518055, P. R. China.}}
\thanks{\IEEEauthorblockA{\IEEEauthorrefmark{5}School of Automobile, Chang'an University, Xi'an, 710064, P. R. China (e-mail: wbyang@chd.edu.cn).}}
\thanks{\IEEEauthorblockA{\IEEEauthorrefmark{6}State Key Laboratory of Grassland Agro-Ecosystem, College of Ecology, Lanzhou University, Lanzhou, 730000, P. R. China (e-mail: shangzhh@lzu.edu.cn).}}
}
\maketitle

\begin{abstract}
Grassland monitoring is essential for the sustainable development of grassland resources. Traditional Internet of Things~(IoT) devices generate critical ecological data, making data loss unacceptable, but the harsh environment complicates data collection. Unmanned Aerial Vehicle (UAV) and mobile edge computing (MEC) offer efficient data collection solutions, enhancing performance on resource-limited mobile devices.
In this context, this paper is the first to investigate a UAV-enabled time-sensitive data collection problem (TSDCMP) within grassland monitoring edge networks~(GMENs). Unlike many existing data collection scenarios, this problem has three key challenges. First, the total amount of data collected depends significantly on the data collection duration and arrival time of UAV at each access point~(AP). Second, the volume of data at different APs varies among regions due to differences in monitoring objects and vegetation coverage. Third, the service requests time and locations from APs are often not adjacent topologically. To address these issues, We formulate the TSDCMP for UAV-enabled GMENs as a mixed-integer programming model in a single trip. This model considers constraints such as the limited energy of UAV, the coupled routing and time scheduling, and the state of APs and UAV arrival time. Subsequently, we propose a novel cooperative heuristic algorithm based on temporal-spatial correlations (CHTSC) that integrates a modified dynamic programming (MDP) into an iterated local search to solve the TSDCMP for UAV-enabled GMENs. This approach fully takes into account the temporal and spatial relationships between consecutive service requests from APs. Systematic simulation studies demonstrate that the mixed-integer programming model effectively represents the TSDCMP within UAV-enabled GMENs. Moreover, the proposed CHTSC algorithm outperforms two superior algorithms across twelve different scale instances.

\end{abstract}

\vspace{-0.17in}
%
\IEEEpeerreviewmaketitle
\section{Introduction}
%
%
%
%
%
%

%
%
Grasslands are among the largest terrestrial biomes, covering more than $25\%$ of the terrestrial earth's surface, and serve as significant hotspots of biodiversity in various regions~\cite{torok2021present}. These ecosystems have very high conservation value and provide essential ecosystem services, including food and fodder, water regulation and supply, erosion control, pollinator promotion, and carbon sequestration~\cite{dengler2014biodiversity}. However, grassland biodiversity and the ecosystem services they offer are severely threatened, making their conservation and restoration top priorities~\cite{andrade2015grassland}. Grassland monitoring is crucial for the scientific utilization and rational development of grassland resource, as well as for maintaining ecological balance, which is central to effective grassland management~\cite{ge2019estimation}.

Internet of Things~(IoT), as a crucial next-generation information technology, has been widely applied to various aspects of grassland monitoring. It plays an irreplaceable role in obtaining meteorological, soil, and biological data, studying ecological change patterns, evaluating ecological quality, and maintaining the sustainable development of grassland ecosystems~\cite{wang2022grassland}. Nevertheless, unlike other application areas of IoT, grassland IoT demands higher requirements for monitoring depth, scope, and frequency. 
These devices are primarily used for fine online monitoring of grassland respiration, snow thickness, desertification, \ce{CO_2} content, light, gas exchange, leaf area, and microorganisms. Consequently, they generate vast amounts of monitoring data which are of great significance for exploring the ecological characteristics of grasslands over time, making it imperative to avoid any data loss.
Moreover, the complex and harsh ecological environment of grasslands, the extensive monitoring range, and the lack of information technology facilities pose significant challenges in deploying IoT on grasslands. These factors result in high data collection costs, significant network transmission delays, and substantial equipment energy consumption. Therefore, designing a feasible network and collecting extensive grassland monitoring data have become significant challenges.

Fortunately, Unmanned Aerial Vehicle~(UAV) offers new opportunities for data collection in IoT networks due to their high agility, flexibility, and low cost~\cite{messaoudi2023survey}. Moreover, mobile edge computing~(MEC) has emerged as a widely studied technology in recent years, demonstrating excellent performance in computation-intensive and latency-critical applications on resource-limited mobile devices~\cite{ning2023mobile}. Inspired by these advancements, this work employs a three-tier network architecture consisting of sensor nodes, access points (APs), and a UAV to collect  grassland monitoring data. We consider a scenario in which a UAV with limited energy is periodically dispatched to collect monitoring data from each AP only once in the monitored grassland area. The UAV carrying the collected monitoring data must return to the base station before depleting its energy.

Due to the heterogeneous functionality of sensors used for grassland monitoring, the sensor nodes may collect varying amounts of data at each AP. This can result in data overflow at certain APs where the collected data exceeds their capacity. What is more, data collection requests are sent to the UAV from APs distributed across various locations and at diverse time intervals as runtime progresses. A UAV with limited energy may not be able to collect all data from each AP sequentially during a single trip within a grassland monitoring region. Therefore, it becomes crucial to determine which APs' data should be prioritized for collection by the UAV? Questions arise about the optimal order and duration for data collection at each AP to maximize the total volume of collected data while minimizing or avoiding data overflow. In other words, the primary question is how the UAV incorporates both data collection deadlines (i.e., data overflow time) and diverse data collection locations into its decision-making process.

Motivated by these issues, this paper investigates the time-sensitive data collection maximization problem~(TSDCMP) within UAV-enabled grassland monitoring edge networks~(GMENs). To the best of our knowledge, no existing work explores such models specifically tailored for UAV-enabled GMENs. Consequently, we formulate the TSDCMP in the context of UAV-enabled GMENs as a mixed-integer programming model. This model accounts for constrains including the limited energy of UAV, the coupled routing and time scheduling, and the state of APs and UAV arrival time. Further analysis of the problem reveals several key features. Firstly, the total amount of data collected is highly dependent on both the data collection duration and the arrival time at each AP. Secondly, the volume of data at each AP typically varies among regions due to differences in monitoring objects and vegetation coverage. Thirdly, the time and location of service requests from two APs are often not adjacent topologically.

Taking the above features and challenges into consideration, this work proposes a novel cooperative heuristic algorithm
based on temporal-spatial correlations (CHTSC) to solve the TSDCMP for UAV-enabled GMENs. The core idea of CHTSC algorithm is to jointly optimize the flight trajectory, arrival time, and data collection duration for the UAV at each AP. This approach fully considers the temporal and spatial relationships between consecutive service requests from APs.
The overarching goal is to minimize the occurrence of data overflow while simultaneously maximizing the volume of data collected from APs, all within the realistic constraints.


The main contributions of this work are summarized as follows:
\begin{enumerate}
  \item We first present the TSDCMP within UAV-enabled GMENs and then formulate a mixed-integer programming model for the TSDCMP under realistic constraints.
  \item We propose a novel CHTSC algorithm to solve the TSDCMP for UAV-enabled GMENs, which fully considers the temporal and spatial relationships between consecutive service requests from APs.
  \item Simulation studies demonstrate that the mixed-integer programming model effectively represents the TSDCMP within UAV-enabled GMENs. Moreover, the proposed CHTSC algorithm outperforms two superior algorithms across twelve different scale instances.
\end{enumerate}

The rest of this work is structured as follows. Section \ref{Related-Work} reviews related work on the UAV-enabled time-sensitive data collection and relevant solutions. Section \ref{System-Model} presents the system model and problem formulation. Section \ref{Cooperative-Heuristic-Algorithm} details the proposed algorithm for optimizing data collection. Section \ref{Simulation-Studies} discusses the experimental setup and results. Finally, Section \ref{Conclusion} concludes this paper and suggests directions for future research.
\vspace{-0.10in}
\section{Related Work}\label{Related-Work}
This work focuses on two aspects: the modeling of UAV-enabled time-sensitive data collection for grassland monitoring and the development of a CHTSC. Therefore, the related literature is reviewed from the these perspectives.

In recent year, UAVs have increasingly been utilized for time-sensitive data collection in various practical scenarios. Samir~\emph{et~al.}~\cite{samir2019uav} investigate the maximizing the number of served IoT devices by jointly optimizing the trajectory of a UAV and the radio resource allocation under a data upload deadline. In addition, Zhan~\emph{et~al.}~\cite{zhan2021multi} design a novel framework for multi-UAV-enabled mobile-edge-computing~(MEC) to provide flexible computational assistance to IoT devices with strict deadlines. Wang~\emph{et~al.}~\cite{wang2020priority} address a priority-oriented trajectory planning problem for a UAV-aided time-sensitive heterogeneous IoT network, providing a solution that ensures the network's latency tolerance within a specific time period.
Liu~\emph{et~al.}~\cite{liu2021uav} tackle the UAV-enabled data collection problem to ensure high information freshness in wireless sensor networks. The information freshness is measured by the age of information~(AoI) for each sensor node.
Tran~\emph{et~al.}~\cite{tran2022uav} propose a time-sensitive data collection scheme using a UAV relay-assisted IoT model. This scheme considers the latency requirement for both uplink and downlink channels to enhance the freshness of information. Liu~\emph{et~al.}~\cite{liu2022uav} investigate the UAV trajectory planning problem in a UAV-enabled environmental monitoring system, focusing on the AoI limitations of data in monitoring areas.
Zeng~\emph{et~al.}~\cite{zeng2023aoi} propose a relay-assisted UAV data collection method that takes into account the freshness of collected data in the large-scale disaster area. Cao~\emph{et~al.}~\cite{cao2023energy} consider a priority-oriented UAV-aided data collection problem in a time-sensitive IoT network with movable sensor nodes. Liu~\emph{et~al.}~\cite{liu2024learning} propose a multi-UAV assisted wireless power transmission space-air-ground power IoT framework for data acquisition and computation, aiming to minimizing the average AoI of power devices.
Moreover, Messaoudi~\emph{et~al.}~\cite{messaoudi2023survey} provide a comprehensive review of the scenarios and key technologies for UAV-assisted data collection in IoT applications and discuss UAV-based time-sensitive data collection.
However, these works overlook both the spatial and temporal data correlations. This oversight results in excessive energy consumption by UAV due to round-trip data collection and causes data overflow at sensor nodes due to untimely collection.

Meanwhile, relatively little work focuses on UAV data collection based on temporal and spatial correlations. Li~\emph{et~al.}~\cite{li2022blockchain} propose a blockchain-enhanced spatiotemporal data aggregation model to reduce data redundancy in UAV-assisted wireless sensor networks. In this model, the UAV must collect data to recover the original raw data in each cluster. However, this approach may not be feasible and can lead to data overflow in some clusters due to the UAV's limited energy, as there may not be sufficient time for collection. Moreover, Guo~\emph{et~al.}~\cite{guo2021minimizing} introduce a fine-grained spatiotemporal model to minimize redundant sensing data transmissions in energy-harvesting sensor networks. Fattoum~\emph{et~al.}~\cite{fattoum2023adaptive} propose an adaptive sampling rate algorithm that uses the spatiotemporal correlation of sampled data and the residual energy of sensor nodes to reduce data acquisition in periodic wireless sensor networks. Xie~\emph{et~al.}~\cite{xie2020geographical} study a geographical correlation-based RF-data collection protocol.
Xu~\emph{et~al.}~\cite{xu2024collect} investigate the maximizing accumulative utility for collected data problem by determining an optimal data collection trajectory for a UAV with limited energy. This accumulative utility evaluates the quality of data, which is spatiotemporally correlated and gathered from various clusters. Although references~\cite{guo2021minimizing,fattoum2023adaptive,xie2020geographical,xu2024collect} examine the temporal and spatial correlations of data, they do not consider the time-sensitive aspects of UAV data collection, and also difficult to extend these methods to our work.

Furthermore, team orienteering problem (TOP) with time-varying profit~\cite{gunawan2016orienteering} presents another significant research challenge. Yu~\emph{et~al.}~\cite{yu2022team} study a variant of TOP with service and arrival-time-varying profit. They further extend this work to a robust variant of TOP with decreasing profits, assuming uncertain service times at customers~\cite{yu2022robust}. These studies assume that the accumulated data value decrease linearly with arrival time. Wan~\emph{et~al.}~\cite{wan2024deep} propose an attention-based deep reinforcement learning enabled multi-UAV scheduling system for disaster data collection with time-varying data value, where UAVs act as temporary and mobile relays. Subsequently, they also introduce a hybrid heuristics-based multi-UAV route planning method for time-dependent IoT data collection~\cite{wan2024hybrid}. However, the aforementioned work do not account for the specific factors affecting grassland monitoring scenarios, making it challenging to directly apply these methods to address the grassland monitoring data collection problem.
\section{System Model and Problem Formulation} \label{System-Model}
In this section, we first describe the system models, which include the GMENs model, UAV time-sensitive data collection model, and UAV energy consumption model. Subsequently, we formulate the mathematical programming model for TSDCMP within UAV-enabled GMENs. The notations of key parameters used in this work are listed in Table \ref{tab1}.
\begin{table}[htbp]
 \caption{MEANINGS OF THE NOTATIONS}\label{tab1}
 \begin{tabular}{m{30pt}<{\raggedright} m{200pt}<{\raggedright}}
  \toprule
  Notation                & Meaning  \\
  \midrule
 $N$                 & The number of APs.  \\
 $T$                 & The maximum mission duration of UAV. \\
 $E_{max}$           & The battery capacity of UAV.  \\
 $t_i$               & The arrival time at AP $i$.\\
 $d_{ij}$            & The distance between APs $v_i$ and $v_j$. \\
 $T^c_i$             & The amount of data collection time spent at AP $i$. \\
 $T^f_{ij}$          & The flight duration of UAV between any two adjacent APs $i$ and $j$.  \\
 $T^o_i$             & The data overflow time of AP $i$.  \\
 $T^h_i$             & The hovering duration of UAV at APs $i$.  \\
 $\tilde{H}$         & The height between the UAV and the AP.  \\
 $\mathcal{L}$       & The locations of these APs.  \\
 $\mathrm{v}$        & The velocity of UAV.\\
 $\mathcal{T}$       & The set of time slots.  \\
 $\mathcal{T}^f$     & The set of time slots for the UAV¡¯s flight state.  \\
 $\mathcal{T}^c$     & The set of time slots for the UAV¡¯s data collection state.  \\
 $R$                 & The achievable transmission rate between AP and UAV follows the free-space path loss model.  \\
 $D_i(t)$            & The data volume of AP $i$ in time slot $t$.  \\
 $C_i$               & The total amount of data collected from AP $i$ by the UAV at the location $l_i$ during the data collection duration $T^c_i$.  \\
 $D_{max}$           & The maximum capacity of each AP.  \\
 $D_{th}$          & The threshold of each AP that requires data collection.  \\
 $\alpha^{s}_i(t)$   & The self-growth rate of data for AP $i$ in time slot $t$.  \\
 $\alpha^{r}_i(t)$   & The real-growth rate of data for AP $i$ in time slot $t$.  \\
 $O$                 & The total amount of data overflowed by the UAV during the data collection tour.  \\
 $C$                 & The total amount of data collected by the UAV during the data collection tour.  \\
  \bottomrule
 \end{tabular}
\end{table}
\subsection{Grassland Monitoring Edge Networks Model}
In this work, the GMENS utilize a three-tier network architecture consisting of sensor nodes, APs, and a UAV to collect grassland monitoring data, as illustrated in Fig.~\ref{Data-Collection-model}. The heterogeneous sensors are deployed across the grassland region, constantly gathering data on various aspects of vegetation growth, including meteorology, soil conditions, microbe activity, soil respiration, carbon dioxide levels, and leaf gas exchange \cite{wang2020alpine}. Each homogeneous AP equipped with a limited buffer and an edge sever processes the raw-data collected from the nearby sensor nodes and stores the processed metadata. Due to the continuous monitoring performed by the sensors, APs need to store substantial volume of data. When the data volume at an AP reaches a certain threshold, it sends a data collection request to the UAV. The UAV with limited energy is periodically dispatched to collect data from $N$ APs within the monitored grassland area. It is assumed that the UAV's storage capacity is sufficient to meet the data collection requirements. The UAV can access information from each AP, including location, data volume, data collection rate, and data collection status, at any time and from any location. The base station handles various tasks, such as dispatching the UAV, charging its batteries, and processing the collected monitoring data.
\begin{figure}[htb]
\begin{center}
$\begin{array}{l}
\includegraphics[width=2.8in]{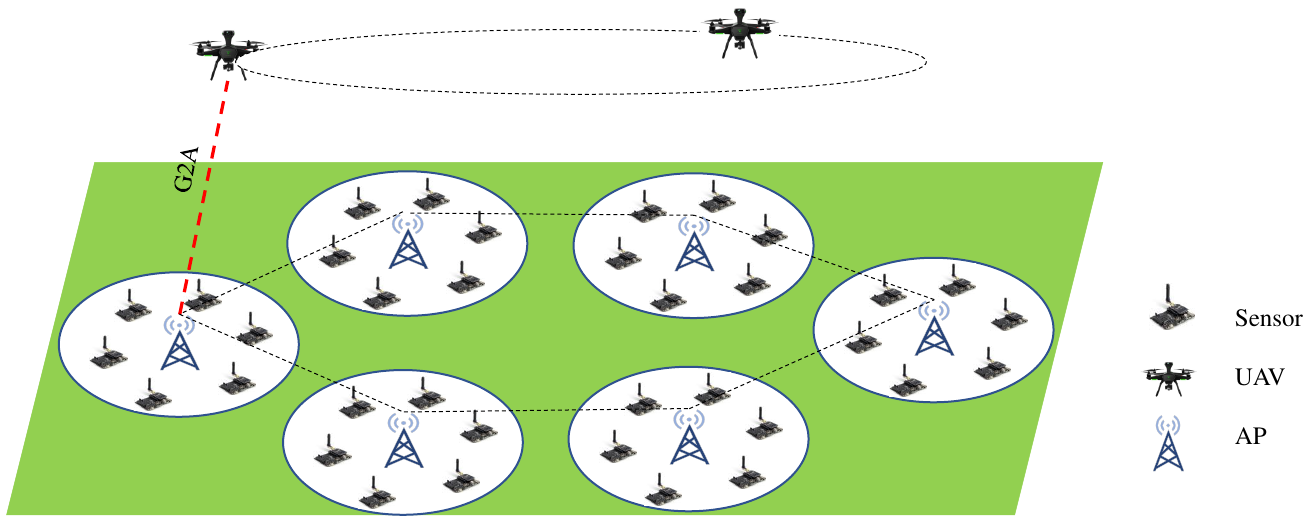}\\
 \end{array}$
\end{center}
\vspace{-0.15in}
\caption{Three-tier network architecture for time-sensitive data collection by a UAV in GMENs.} \label{Data-Collection-model}
\end{figure}

Let $\mathcal{N}_a = \{v_1,\ldots, v_N\}$  denote the set of APs. The locations of these APs are represented as $\mathcal{L} = \{l_i = (x_i, y_i, 0)| 1\leq i \leq N\}$. Due to the uneven terrain of the grassland, the altitude of the UAV relative to the ground is fixed at constant for safety reasons. The UAV collects data from each AP only once and subsequently delivers it to the data center for further processing upon returning to the base station $v_0$  before exhausting its energy. Therefore, the UAV's trajectory for data collection must form a closed tour. The GMENS can be described as a complete directed weighted graph $G = (\mathcal{N}, A)$, where $\mathcal{N} = \{v_0\}\cup \mathcal{N}_a$ represents the set comprising all APs and the base station $v_0$, and $A = \{(v_i,v_j)|v_i,v_j \in \mathcal{N}, i\neq j\}$ denotes the set of the paths connecting any two APs. Each arc $a \in A, a = (v_i,v_j)$ has an associated non-negative weight value $d_{ij}$ that indicates the distance between APs $v_i$ and $v_j$, i.e., $d_{ij} = \parallel l_i -l_j\parallel$. For convenience, we interchangeably refer to AP $v_i$ as AP $i$ or $i$ in the remainder of this paper.

The maximum mission duration $T$ of UAV is divided into equal time slots, with each time slot having a length of $\tau$. We assume that $\tau$ is sufficiently small so that the UAV's location change within $\tau$ is negligible. Let $\mathcal{T} = \{t | 1 \leq t \leq T \}$ denote the set of time slots.

To fully utilize solar energy for power supply, the height of APs is typically set higher than the surface of grassland vegetation. Consequently, the wireless channel between the UAV and AP is dominated by the line of sight (LoS) link. In addition, considering the data collection rates and the energy consumption associated with transmission distance, the data can be collected from an AP only when the UAV hovers directly over it. This operational mode is achieved by following the fly-hover-communication protocol \cite{zeng2019energy}. The achievable transmission rate $R$ between AP and UAV follows the free-space path loss model. It can be defined by
\begin{equation} \label{total-amount-data}
\begin{small}
R = B\log_2(1+\frac{p \beta}{\sigma^2\tilde{H}^2}),
\end{small}
\end{equation}
where $B$ presents the channel bandwidth, $p$ is the transmit power of AP for uploading data link, $\beta$ denotes the channel power gain when the reference distance is set as 1 m, $\sigma^2$ denotes the noise power, and $\tilde{H}$ is the height between the UAV and the AP, where $\tilde{H}$ is greater than the height from the UAV to the ground.
\subsection{Time-Sensitive Data Collection Model}
To reduce the energy consumption of the UAV flying back and forth, we assume that the UAV collects data from each AP only once. The data collection duration of UAV at each AP can be expressed as a set $\mathcal{T}^c = \{T^c_i| T^c_i\geq 0, l_i \in \mathcal{L}\}$ in term of time slots. Meanwhile, the flight duration of UAV between any two adjacent APs $i$ and $j$ can be expressed as a set $\mathcal{T}^f = \{T^f_{ij}| T^f_{ij}\geq 0, l_i, l_j \in \mathcal{L}\}$ in term of time slots.

Assume that the UAV keeps a fixed velocity $\mathrm{v}$. Thus the flight duration of UAV between any two adjacent APs $i$ and $j$ can be calculated by
\begin{equation} \label{fight-duration}
\begin{small}
T^f_{ij} = \frac{\parallel l_i -l_j\parallel}{\mathrm{v}\tau}.
\end{small}
\end{equation}

To ensure that the total duration of data collection and flight for the UAV does not exceed $T$ time slots, we impose the following constraint:
\begin{equation} \label{total-data-collection-duration}
\begin{small}
\sum_{T^c_i \in \mathcal{T}^c} T^c_i + \sum_{T^f_{ij} \in \mathcal{T}^f} T^f_{ij} \leq T.
\end{small}
\end{equation}

The total amount of data collected from AP $i$ by the UAV at the location $l_i$ during the data collection duration $T^c_i$ is denoted as $C_i$, and it can be calculated as
\begin{equation} \label{total-amount-data}
C_i = \tau T^c_i R.
\end{equation}

Note that the data volume of each AP dynamically grows over time during the UAV's data collection process. Let $\alpha^{s}_i(t)$ denote the self-growth rate of data for AP $i$ in time slot $t$. This rate is constant for a given AP but may vary between different APs due to the heterogeneity of the surrounding sensor nodes. The real-growth rate $\alpha^{r}_i(t)$ of data in time slot $t$ can be given as
\begin{equation} \label{data-rate-AP}
\begin{small}
    \alpha^{r}_i(t) =
   \begin{cases}
   \alpha^{s}_i(t), &\mbox{$t \in \mathcal{T}^f$} \\
   \alpha^{s}_i(t) - R , &\mbox{$t \in \mathcal{T}^c$},
   \end{cases}
\end{small}
\end{equation}
where $\mathcal{T}^f$ and $\mathcal{T}^c$ denote the set of time slots for the UAV's flight state and data collection state, respectively. Eq. \eqref{data-rate-AP} means that the real-growth rate of data for AP $i$ is $\alpha^{s}_i(t)$ when the UAV is in flight state, and otherwise it equals the difference between the self-growth rate of data $\alpha^{s}_i(t)$ for AP $i$ and the data collection rate $R$ by the UAV when it collects data from AP $i$. Therefore, the data volume of AP $i$ in time slot $t$ can be calculated as follows:
\begin{equation} \label{data-volume-AP}
D_i(t) = D_i(t-1) + \alpha^{r}_i(t)\tau.
\end{equation}

Each AP operates continuously, and when the monitoring data collected from neighbouring sensors exceeds its maximum capacity $D_{max}$, it results in data overflow, leading to the loss of critical ecological information about the grassland. Moreover, the varying volumes of data collected by each AP from neighboring sensors result in different overflow time for each AP. To avoid data overflow at APs with excessive data collected, the UAV must collect data in real-time based on the service requests of each AP.
Therefore, we set a threshold $D_{th}$ for each AP. When the data volume of AP $i$ exceeds this threshold, the AP will immediately send a data collection request to the UAV, which will then schedule a response. Conversely, if the data volume of AP $i$  remains below the threshold $D_{th}$, data overflow will not occur. Consequently, the volume of data collected by the UAV at AP $i$ satisfies
\begin{equation} \label{data-volume-AP}
D_i-D_{th} \leq C_i \leq D_i.
\end{equation}

\begin{figure}[htb]
\begin{center}
$\begin{array}{l}
\includegraphics[width=2.6in]{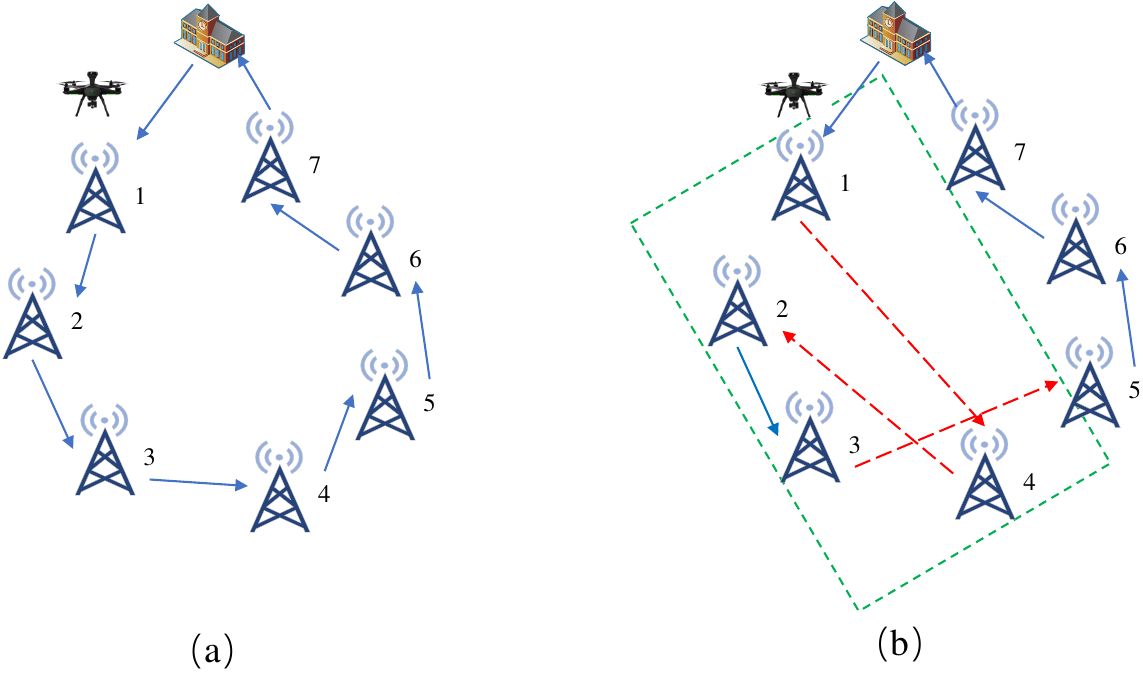}\\
 \end{array}$
\end{center}
\vspace{-0.15in}
\caption{Example of real-time data collection by a UAV in GMENs: (a) The shortest distance Hamiltonian cycle. (b) The path of real-time data collection.} \label{Real-time-data-collection}
\end{figure}
The real-time data collection scheme by the UAV is illustrated in Fig. \ref{Real-time-data-collection}.  Fig. \ref{Real-time-data-collection}a depicts the shortest distance Hamiltonian cycle formed by seven APs based on the distances between them. Fig. \ref{Real-time-data-collection}b shows the path of real-time data collection in response to the data collection requests from the APs. The data collection request signals sent from each AP are first stored in a response queue $Q$ of the UAV in the order of their request service time, as shown in Fig. \ref{Queueing model}. Assume that $Q$ is large enough to store all request signals. There is an unequal time interval between two request signals before and after. Subsequently, the UAV successively collects the monitoring data from each AP based on the sequence in the queue $Q$.
\begin{figure}[htb]
\begin{center}
$\begin{array}{l}
\includegraphics[width=2.1in]{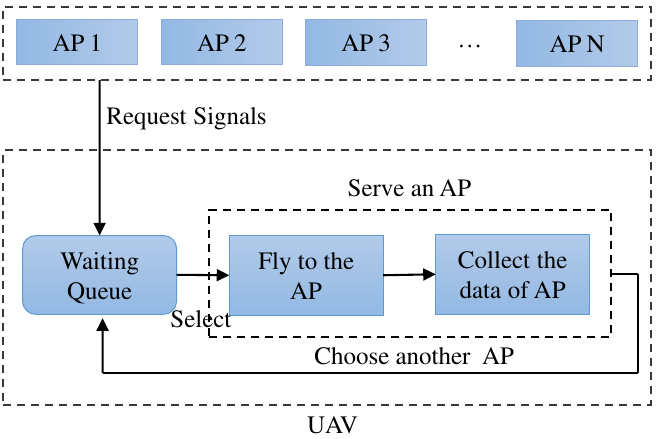}\\
 \end{array}$
\end{center}
\vspace{-0.15in}
\caption{Queueing model for real-time data collection scheme by a UAV.} \label{Queueing model}
\end{figure}

Due to the different priorities of request service for each AP, the UAV should collect data from each AP before it overflows. We define the data overflow time of AP $i$ as $T^o_i$, which can be calculated as follows:
\begin{equation} \label{data-overflow-time}
\begin{small}
T^o_i = \frac{D_{max} - D_{th}}{\alpha^{s}_i}.
\end{small}
\end{equation}

According to Eq. \eqref{fight-duration}, the flight duration of UAV between any two APs $i$ and $j$ in the response queue $Q$ is $T^f_{ij}$. Hence, for the UAV to collect data from AP $j$ before its data overflows, the time of data collection at AP $j$ must satisfy
\begin{equation} \label{total-amount-data}
\begin{small}
T^c_j = t + T^c_i + T^f_{ij} \leq T^o_j.
\end{small}
\end{equation}
Here the data collection delay is ignored compared to the time required for data collection and flight.
\subsection{Energy Consumption Model for UAV}
The rotary-wing UAVs typically operate for short durations due to the limited energy capacity of their batteries. Moreover, the battery lifetime of a UAV is influenced by several factors, including battery size, speed, weight, and environmental conditions. In general, UAV energy consumption primarily consists of two components: propulsion energy and communication energy.

Without loss of generality, the communication energy is negligible compared with the propulsion energy and can be ignored in this work. For simplicity, we also disregard the additional energy consumption caused by the UAV acceleration or deceleration. This assumption is reasonable for the real-time data collection scenario in this work since the UAV maneuvering time constitutes only a small fraction of the total operation time. For propulsion energy, the propulsion power consumption model for rotary-wing UAVs is adopted in this work \cite{zeng2019energy}. It is expressed as follows:
\begin{equation} \label{propulsion-power-consumption-model}
\begin{small}
\begin{aligned}
P(\mathrm{v}) = & P_0\left(1+\frac{3\mathrm{v}^2}{U^2_{tip}}\right)+P_i\left(\sqrt{1+\frac{\mathrm{v}^4}{4\mathrm{v}^4_0}}-\frac{\mathrm{v}^2}{2\mathrm{v}^2_0}\right)^{\frac{1}{2}} \\
&+\frac{1}{2}d_0\rho sA\mathrm{v}^3,
\end{aligned}
\end{small}
\end{equation}
where $P_o$ and $P_i$ are the blade profile power and induced power in hovering status, respectively, $U^2_{tip}$ is the tip speed of the rotor blade, $\mathrm{v}_0$ is the mean rotor induced speed in hover, and $d_0$ and $s$ are the fuselage drag ration and rotor solidity, respectively. Also, $\rho$ and $A$ are the air density and rotor disc area, respectively.

For exposition purposes, the hovering duration and energy consumption at the base station are assumed to be zero, namely $T_0 = 0$ and $E_0 = 0$. Thus, the flight duration on the arc $(i,j)$ can be denoted as
\begin{equation} \label{flight-duration-arc}
\begin{aligned}
   T_{ij} = &T^h_i + T^f_{ij} \\ = &
   \begin{cases}
   \frac{\parallel l_i-l_j\parallel}{\mathrm{v}}, &\mbox{$i = 0, (i,j) \in A$} \\
   T^c_i + \frac{\parallel l_i-l_j\parallel}{\mathrm{v}}, &\mbox{$i \neq 0, (i,j) \in A$},
   \end{cases}
\end{aligned}
\end{equation}
where the UAV hovering duration consists of the data collection duration $\mathcal{T}^c$ and the hovering duration at the base station $T_0$, i.e., $\mathcal{T}^h = \mathcal{T}^c \cup \{T_0\}$.

Correspondingly, the energy consumption on the arc $(i,j)$ can be denoted as
\begin{equation} \label{energy-consumption-arc}
\begin{aligned}
   E_{ij} = &E^h_i + E^f_{ij} \\ = &
   \begin{cases}
   P^f\frac{\parallel l_i-l_j\parallel}{\mathrm{v}}, &\mbox{$i = 0, (i,j) \in A$} \\
   P^c T^c_i + P^f\frac{\parallel l_i-l_j\parallel}{\mathrm{v}}, &\mbox{$i \neq 0, (i,j) \in A$},
   \end{cases}
\end{aligned}
\end{equation}
where $P^c$ denotes the power consumption of UAV during data collection in hovering status with the velocity $\mathrm{v}=0$ and then $P^c = P_0 + P^c_i$.
\vspace{-0.10in}

\subsection{Problem Formation}\label{Problem-Formation}
The goal of the TSDCMP within UAV-enabled GMENs is to maximize the total amount of data collected while minimizing the total amount of data overflowed during the data collection tour within the limited energy capacity $E_{max}$.
To achieve this goal, the problem involves determining a single route for the UAV that departs from and returns to the base station $v_0$. Furthermore, it needs to allocate the appropriate duration of data collection at each visited AP.
In the following sections, we first detail the related notations necessary to solve the problem, and then formulate a mixed-integer programming solution for it.

Let $S(t_i)$ denote the state of AP $i$ when the UAV arrives at it at time $t_i$. There are two possible states, i.e., data collection state $S(t_i) = 0$ and data overflow state $S(t_i) = 1$, satisfying
\begin{equation}
\begin{small}
    S(t_i) =
   \begin{cases}
   0, &\mbox{$t_i \leq T^o_i$} \\
   1, &\mbox{$t_i > T^o_i$},
   \end{cases}
\end{small}
\end{equation}
where $t_i$ and $T^o_i$ are the arrival time and the data overflow time for AP $i$, respectively. Then, the total amount of data overflowed $O$ for all APs can be expressed as
\begin{equation} \label{total-data-overflow}
\begin{small}
 O = \sum^N_{i=1} S(t_i)(t-T^o_i)\tau R.
\end{small}
\end{equation}

Meanwhile, the total amount of data collected by the UAV during the data collection tour,  eventually gathered from all APs, can be expressed as
\begin{equation} \label{total-amount-data-collection}
\begin{small}
\begin{aligned}
C = &\sum^N_{i=1} ((1-S(t_i)C_i)+S(t_i)\bar{C}_i) \\
  = &\sum^N_{i=1} ((1-S(t_i)\tau T^c_i R)+ S(t_i)\tau \bar{T}^c_i R),
\end{aligned}
\end{small}
\end{equation}
where $\bar{C}_i$ and $\bar{T}^c_i$ respectively denote the amount of data collected and the data collection duration  by the UAV after the data overflow at AP $i$. Eq. \eqref{total-amount-data-collection} means that when the AP $i$ is in the data collection state at arrival time $t_i$, i.e., $S(t_i) = 0$, the amount of data collected by UAV at time $t_i$ is $C_i(1-S(t_i))$, otherwise, $C_i = \bar{C}_iS(t_i)$.

To further understand and solve the TSDCMP within UAV-enabled GMENs, we propose a mixed-integer programming formulation for the problem. This formulation can be expressed as follows:
\begin{subequations}\label{data-collection-maximization-problem-mathematical-model}
\begin{align}
\max \quad &C - \delta O 
\label{maximizing-data-collection}\\
\emph{s.t.}\quad &\sum^N_{i=1} P(0)T^c_i + \sum^N_{i=0} \sum^N_{j\neq i}x_{ij}P(\mathrm{v})T^f_{ij} \leq E_{max}, \forall i,j\in \mathcal{N}_a,  \label{total-time-constraint}\\
&t_i + T^c_i + T^f_{ij} -t_j \leq M(1-x_{ij}), \forall i,j\in \mathcal{N}_a, \label{couple-routing-decision-with-time} \\
&t_i > T^o_i - M(1-S(t_i)), \forall i \in \mathcal{N}_a, \label{arrival-time-constraint-a} \\
&t_i \leq T^o_i + MS(t_i), \forall i \in \mathcal{N}_a, \label{arrival-time-constraint-b} \\
&D_i-D_{th} \leq C_i \leq D_i, \forall i \in \mathcal{N}_a, \label{data-collection-constraint-AP}\\
&\sum^N_{j=1,j\neq k}x_{kj} = y_k, \forall k \in \mathcal{N}_a, \label{indicator-variable-x-y}\\
&\sum^N_{i = 0,i\neq k}x_{ik} = \sum^{N}_{j =0,j\neq k}x_{kj}, \label{flow-conservation-condition}\\
&\sum^N_{j = 1}x_{0j} = \sum^N_{j =1}x_{j0} = 1, \label{start-end-condition}\\
&S(t_i) \in \{0,1\}, \forall i \in \mathcal{N}_a, \label{variable-s}\\
&x_{ij} \in \{0,1\}, \forall i,j \in \mathcal{N}_a, \label{variable-x}\\
&y_i \in \{0,1\}, \forall i \in \mathcal{N}_a, \label{variable-y}\\
&T^c_i, t_i \geq 0, \forall i \in \mathcal{N}_a, \label{variable-time}
\end{align}
\end{subequations}
where $\delta$ is the penalty parameter; the binary variables $x_{ij} =1 $ if the UAV flies directly from AP $i$ to $j$, and $0$ otherwise; the binary variables $y_i = 1$ if AP $i$ is collected data by the UAV, and $0$ otherwise, which is an auxiliary decision variable that is uniquely determined by variables $x_{ij}$. 

The objective function \eqref{maximizing-data-collection} is to maximize the total amount of data collected while minimizing the total amount of data overflowed for all APs by the UAV on the data collection tour, where the total amount of data collected $C$ by the UAV defined in Eq. \eqref{total-amount-data-collection} and the total amount of data overflow for all APs defined in Eq. \eqref{total-data-overflow}. Constraints \eqref{total-time-constraint} ensure that the total energy consumption of both data collection and flight in a scheduling cycle of UAV cannot exceed the total energy capacity $E_{max}$. Constraints \eqref{couple-routing-decision-with-time} couple routing decisions with the time scheduling, where $M$ is a large positive constant \cite{cacchiani2020models}. The binary variables $x_{ij} =1$ if the UAV arrives at AP $j$ directly after departing from AP $i$, the arrival time at AP $j$ must be equal to or greater than the sum of both arrival time  and data collection time of AP $i$, and flight time between APs $i$ and $j$. Constraints \eqref{arrival-time-constraint-a} and \eqref{arrival-time-constraint-b} denote the state of AP $i$ when the UAV arrives at it at time $t_i$. Constraints \eqref{data-collection-constraint-AP} ensure that the amount of data collection by the UAV at each AP is within a safe range. Constraints \eqref{indicator-variable-x-y} establish the relationship between variables $x$ and $y$ and ensure that each AP can be collected data at most once. Constraints \eqref{flow-conservation-condition} are the flow conservation for each AP. Constraints \eqref{start-end-condition} ensure that the UAV begins and ends its route at the base station. Constraints \eqref{variable-s} show the state indicator variables at arrival time $t_i$.
Constraints \eqref{variable-x}--\eqref{variable-time} define the domain of the decision variables.

Further analysis of the optimization problem \eqref{data-collection-maximization-problem-mathematical-model}, we find that maximizing the data collection efficiency of the UAV inherently involves minimizing the total amount of data overflowed during the data collection tour. Furthermore, the total amount of data collected and data overflowed depends on the duration of data collection and the arrival time at each AP during the data collection tour under the limited total energy capacity $E_{max}$. Minimizing the volume of overflowed data necessitates precise scheduling decisions for the UAV. This often lead to increased  energy consumption for UAV flights, thereby potentially reducing the total amount of data that can be collected by the UAV. Therefore, finding a solution that maximizes the total amount of data collected while minimizing the volume of overflowed data under the UAV's energy constraint is crucial. In the following section, we develop an efficient algorithm based on temporal-spatial correlations to address this problem.
\begin{theorem}
The TSDCMP within UAV-enabled GMENs is NP-hard.
\end{theorem}
\begin{proof}
We demonstrate that the TSDCMP within UAV-enabled GMENs is NP-hard by reducing it from a classic NP-hard problem, namely the TOP with time-varying profit \cite{gunawan2016orienteering,yu2022team}. This reduction shows that the the TSDCMP within UAV-enabled GMENs can be transformed into an TOP in $G$ as follows.

Given a node and an associated arc in complete directed weighted graph $G (\mathcal{N},A)$, where each node $v_i \in \mathcal{N}$ has a positive profit $p_i$, and a profit function $f_i(a_i, s_i)$ varying with the arrival time $a_i$ and the duration of service time $s_i$. Nodes $v_0$ and $n+1$ correspond to the starting and ending the depots, respectively, and do not generate profit. Each node can be visited at most once within a given limited time $T_{max}$. The aim of the TOP with time-varying profit is to determine a route that visits a subset of vertexes, and to determine the duration of the arrival time and service time spent on each visited vertex, so as to maximize the total collected profit under the given limited time $T_{max}$.

We reduce the TOP with time-varying profit in $G (\mathcal{N},A)$ to the TSDCMP within UAV-enabled GMENs as follows. The profit collected at each AP $i \in \mathcal{N}$ is $C_i$ with the duration of service time $T_i \in \mathcal{T}^c$, and the arrival time $a_i$ needs to be as early as possible to meet the data overflow time $T^o_i$. The volume of data stored $D_i$ at each AP varies significantly, which means that the hovering duration $T_i$ of the UAV for data collection at each AP is different if the data volume of each AP after data collection is at least below its threshold $D_{th}$. Consequently, the amount of data collected $C$ by the UAV at each AP also varies. Finally, the total amount of data collected by the UAV is different in each data collection mission. The overall objective is to collect as much data as possible while ensuring that the data collection duration of UAV is within the mission duration $T$. 
 Since the TOP with time-varying profit is NP-hard \cite{gunawan2016orienteering,yu2022team}, the TSDCMP within UAV-enabled GMENs is NP-hard.
\end{proof}
\section{Cooperative Heuristic Algorithm Based on Temporal-Spatial Correlations}\label{Cooperative-Heuristic-Algorithm}
In this section, we propose a novel CHTSC to solve the TSDCMP within UAV-enabled GMENs. The problem can be partitioned into two subproblems, namely the trajectory planning and the data collection duration scheduling. The former determines the flight trajectory of UAV during data collection, while the latter focuses on allocating the arrival time and data collection duration at each AP. Subsequently, an iterated local search is adopted to plan the trajectory of UAV. Given the vast number of trajectories generated by iterated local search, it is particularly crucial to quickly solve the UAV data collection duration scheduling subproblem. To address this, we design an efficient algorithm modified from the dynamic programming (DP) algorithm.

In the following paper, we detail the CHTSC called the iterated local search with modified DP (ILS-MDP) to tackle the TSDCMP within UAV-enabled GMENs.
\vspace{-0.15in}
\subsection{Framework of ILS-MDP for TSDCMP}
ILS-MDP is a hybrid heuristic algorithm that aptly solves the TSDCMP within UAV-enabled GMENs by leveraging the cooperation between ILS and MDP. The framework of ILS-MDP for solving the problem is illustrated in Algorithm \ref{data-collection-temporal-spatial-framework}. The ILS comprises perturbation and local search. At the beginning of ILS, we initialize a feasible data collection scheme $S = \{(V,T^c)\}$ and then perform perturbation and local search on $S$. In each iteration, a perturbation is employed to generate an improved UAV flight trajectory. Subsequently, the solution is explored and improved by the local search, which consists of several neighborbood structures. In each neighborbood structure, the search operators such as \emph{or-opt, swap, and 2-opt}~\cite{JIAO2024108084} are applied to the current routing decisions until the best routing is found. The MDP algorithm is utilized to evaluate each improved routing $V$, and to determine the data collection duration $T^c$ as well as the objective function value $f(V,T^c)$. Finally, the algorithm terminates when the optimal solution can no longer be further improved within the termination criteria $T_{max}$.
\begin{algorithm}[htbp]
\begin{small}
\caption{The framework of ILS-MDP for TSDCMP.} \label{data-collection-temporal-spatial-framework}
\KwIn{A feasible route $V$ and the termination criteria $T_{max}$\;}
\KwOut{Optimal time-sensitive data collection scheme $S^*$\;}
Initialize a feasible data collection scheme $S = \{(V,T^c)\}$, let $S^* \leftarrow S$ and $t \leftarrow 0$  by calling Algorithm \ref{Algorithm2}\;
\While{$t < T_{max}$}
{$S'$ $\leftarrow$ \emph{Perturb} $(S^*)$ by calling Algorithm \ref{Perturbation}\;
\If{No new route generated}{\Continue\;$t \leftarrow t+1$\;}
 $S''$ $\leftarrow$ \emph{LocalSearch} ($R$, $S'$, MDP)\;
\eIf{$f(S'')>f(S^*)$}{$S^* \leftarrow S'', t \leftarrow 0$\;}{$t \leftarrow t+1$\;}
}
\end{small}
\end{algorithm}

We present the details on the initialization, perturbation, and local search of ILS-MDP in the following sections.
\subsection{Initialization for ILS-MDP}
Through in-depth investigation on the TSDCMP within UAV-enabled GMENs, we realize that the shortest distance Hamiltonian cycle is not necessarily the best solution and even the infeasible, despite its potential to minimize the UAV's energy data collection energy consumption. Conversely, a more general solution, which may not produce the shortest distance Hamiltonian cycle or minimize energy consumption, could be feasible and potentially optimal for the UAV. This discrepancy arises because the data collection request signals of APs in a response queue $Q$ are sorted by their service request order. This sorting introduces a temporal correlation between adjacent service requests in $Q$. Typically, the APs at the front of $Q$ are prioritized on the data collection tour. To enhance the algorithm's decision-making efficiency, we leverage the data collection deadlines (i.e., data overflow time) of the APs to construct a Hamiltonian cycle. This cycle, referred to as the shortest time Hamiltonian cycle $H$, is illustrated in Fig. \ref{Passer-By-Data-Collection}a.
\begin{figure}[htb]
\begin{center}
$\begin{array}{l}
\includegraphics[width=2.8in]{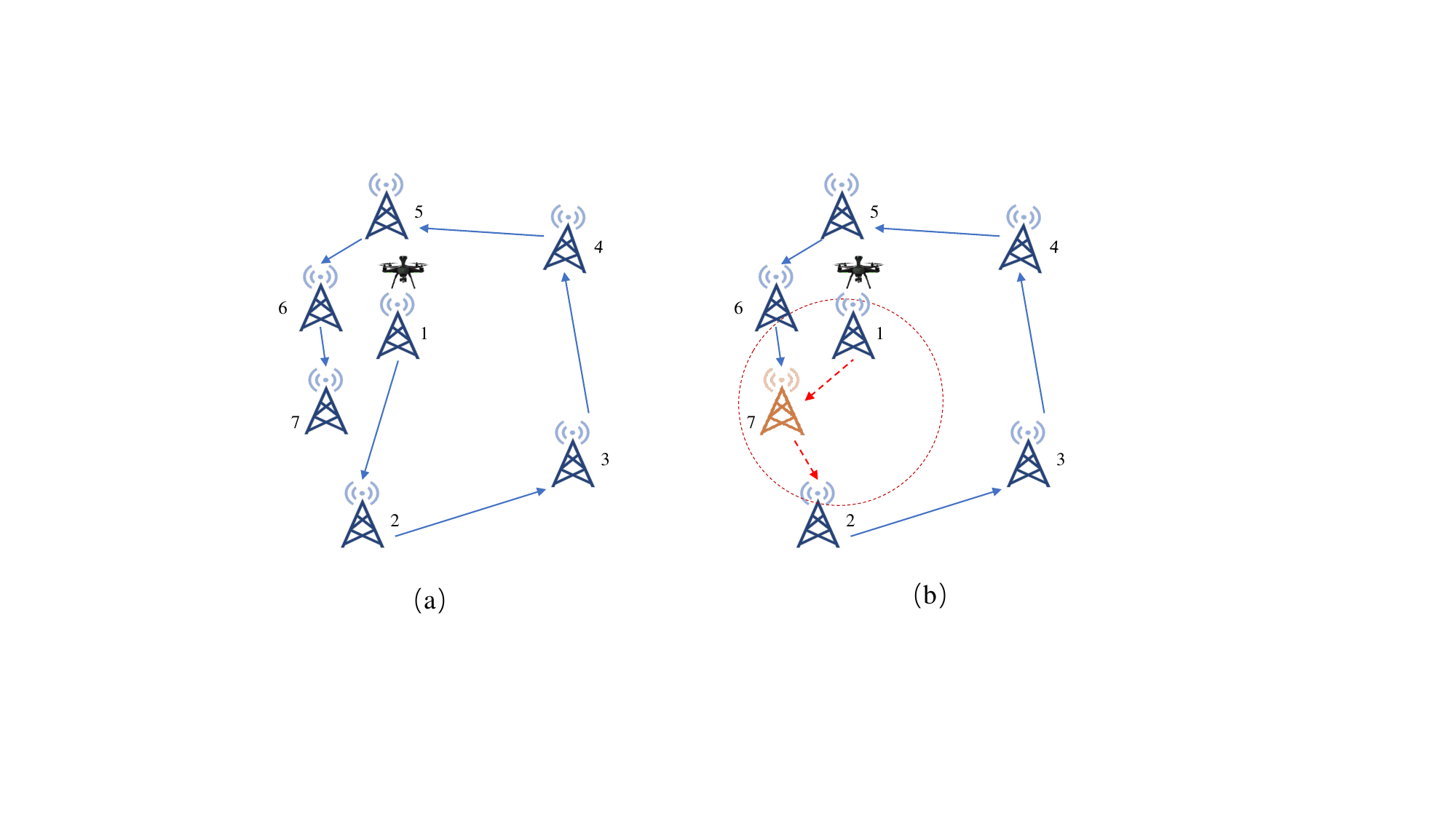}\\
 \end{array}$
\end{center}
\vspace{-0.15in}
\caption{Example of the local search strategy based on temporal-spatial correlations for passer-by APs: (a) The shortest time Hamiltonian cycle. (b) The APs data collection based on temporal-spatial correlations.} \label{Passer-By-Data-Collection}
\end{figure}

On the basic of the shortest time Hamiltonian cycle $H$, we need to select a feasible route for data collection by the UAV. Intuitively, two APs with adjacent priorities of request service may not be spatially adjacent. Therefore, to reduce excessive energy consumption from the UAV flying back and forth \cite{JIAO2024108084}, the route should be optimized within the shortest time hamiltonian cycle $H$, taking into account both the priority of service time and spatial proximity, i.e., \emph{temporal-spatial correlations}. In specific, we optimize the data collection order of APs without violating the UAV's energy constraint to minimize data overflow. This involves ensuring that APs with similar service priorities are as close as possible spatially, thereby reducing energy consumption and maximizing the amount of data collection by the UAV.

In light of the aforementioned discussion, we assign a score to evaluate the order in which each AP is served during the initialization for ILS-MDP. The scoring mechanism fully considers both the time-sensitivity and geographical proximity among APs. More specifically, an AP with an earlier service time and shorter required flight time will have a higher priority to be serviced first. The score $\omega_i$ of AP $i$ is calculated as follows:
\begin{equation} \label{score}
\begin{small}
\omega_i = \frac{1}{i\Delta T_i},
\end{small}
\end{equation}
where $i$ denotes the order in which AP $i$ is served in the shortest time Hamiltonian cycle $H$, and $\Delta T_i$ is the increased time after inserting AP $i$ into the current sequence of APs $V$.

A complete feasible solution $S$ can be denoted as $S = \{(V,T^c)|V=\{v_0, v_1, \ldots, v_n, v_0\},T^c=\{0, T_1^c, \ldots, T_n^c, 0\}\}$ for the shortest time Hamiltonian cycle $H$ constructed by the response sequence $Q$, where $V$ and $T^c$ are the visited sequence of APs and the corresponding data collection duration for these APs, respectively. Firstly, we initialize the visited sequence of APs $V$ consisting only of the BS $v_0$. Secondly, we attemp to select AP $i$ with the highest score $\omega_i$ to be added to the visited sequence $V$ and calculate the arrival time $t_{i^*}$ for AP $i^*$.
If the UAV successfully reaches AP $i^*$ before its data overflow time $T^o_{i^*}$, then AP $i^*$ is inserted into the visited sequence $V$. Subsequently, the data collection duration is calculated and the relevant information is updated. Otherwise, if the UAV does not reach AP $i^*$ before its data overflow time $T^o_{i^*}$, AP $i^*$ is successively inserted into each possible position within the visited sequence $V$. If the status of all APs in $V$ remains unchanged after the insertion of AP $i^*$ into a specific position, the insertion is successfully, and the related information is updated accordingly. If the AP $i^*$ fails to insert $V$, it is added to the end of the visited sequence $V$, and the related information will be updated simultaneously. The detailed process is displayed in Algorithm \ref{data-collection-temporal-spatial}.
\begin{algorithm}[htb]
\begin{small}
\caption{Initialize a feasible data collection scheme.} \label{data-collection-temporal-spatial}\label{Algorithm2}
\KwIn{The set of all APs $\mathcal{N}_a$ and BS $v_0$, the self-growth rate of data for each AP $\alpha^{s}_i(t)$, and the UAV's battery capacity $E_{max}$\;}
\KwOut{A feasible data collection scheme $S$\;}
Generate the shortest time Hamiltonian cycle $H$ based on the response queue $Q$, i.e., $H \leftarrow Q$\;
Let $S = \{(V,T^c)|V=\{v_0,v_0\},T^c=\{0,0\}\}$\;
\While{$H \neq \emptyset$}{
\For{$i \leftarrow 1$ to $|H|$}{$\omega_i \leftarrow \frac{1}{i\Delta T_i}$\;}
$i^* \leftarrow \arg \max_{i\in \mathcal{N}_a} \omega_i$\;
Calculate the arrival time $t_{i^*}$ for $i^*$\;
\eIf{$t_{i^*} \leq T^o_{i^*}$}{
$T^c_{i^*} \leftarrow \frac{D_{i^*}+\alpha^{s}_i t_{i^*} -D_{th}}{R}$\;
$V \leftarrow V \cup \{i^*\}$\;
$H \leftarrow H \backslash \{i^*\}$\;
$T^c \leftarrow T^c \cup \{T^c_{i^*}\}$\;
$S(t_{i^*})\leftarrow 0$\;}{\ForEach{AP $j \in V$}{Insert $i^*$ next to $j$\;
\If{APs in $V$ have no status changes}{
Execute the same operation as lines 10-15 and update the related information for subsequent nodes\;
\textbf{break}\;
}
}
\If{Node $i^*$ insertion failed}
{
$V \leftarrow V \cup \{i^*\}$\;
$H \leftarrow H \backslash \{i^*\}$\;
$T^c \leftarrow T^c \cup \{T^c_{i^*}\}$\;
$S(t_{i^*})\leftarrow 1$\;
}
}
}
\end{small}
\end{algorithm}
\vspace{-0.09in}
\subsection{Perturbation and Local Search Based on Temporal-Spatial Correlations}
\subsubsection{Perturbation}
The perturbation method plays a crucial role as an integral component of ILS-MDP, aimed at enhancing the algorithm's search capabilities by disrupting local optimal solutions. The core idea of the perturbation method is to remove a subset of visited APs and then insert unvisited APs to generate a diversified solution.

We define a perturbation strength parameter $\xi$ to denote the number of APs to be removed from the route. Considering the temporal-spatial correlations between APs, we set $\xi = |V|/10 +1$, where $|V|$ is the number of APs in the current route. Algorithm \ref{Perturbation} describes the perturbation method process. Firstly, the $\xi$ visited APs are randomly selected and removed from the route $V$, and then the removed $\xi$ APs are randomly reinserted into the current route $V$. If all inserting positions satisfy the UAV's energy constraint and yield a higher the objective function value
compared to the original solution, the perturbation of the solution $S$ is considered successful, and solution $S$ is updated with the perturbed solution $S^*$. Otherwise, the perturbation of solution $S$ is deemed unsuccessful. After perturbation, the algorithm enters the local search phase.
\vspace{-0.20in}
\begin{algorithm}[htb]
\begin{small}
\caption{Perturb generate a diversified solution.}\label{Perturbation}
\KwIn{A feasible solution $S$\;}
\KwOut{A feasible new solution $S^*$ generated after disturbance\;}
 Randomly select and remove $\xi$ visited APs from $V$\;
 Randomly reinsert $\xi$ removed APs into $V$\;
  \eIf{$E(S^*) \leq E_{max}$ and value $(S)$ $\leq$ value $(S^*)$}{
  $S \leftarrow S^*$\;
  \Return true;
  }
  {
  No feasible solution $S^*$\;
  \Return false;
  }
\end{small}
\end{algorithm}
\subsubsection{Local Search}
Local search is another key component in ILS-MDP. The basic idea of local search is to start with a current solution and iteratively explore neighboring solutions by employing several neighborbood structures based on the characteristics of the TSDCMP. The neighborbood structure defines the set of neighboring solutions that can be obtained by utilizing the search operators to the current solution. Specific to this problem, a given solution can be improved by adjusting the order of APs' request service to reduce data overflow and increase data collection. In this work, we adopt three classical search operators (i.e., \emph{2-opt}, \emph{swap}, and \emph{or-opt}) for the flight trajectory problem. They are described as follows:
\begin{itemize}
  \item \emph{2-opt}: Remove two non-adjacent edges from the current solution and reconnecting them in a different way to generate a new solution.
  \item \emph{swap}: Exchange the positions of two APs in the current route.
  \item \emph{or-opt}: Remove a subsequence of $\delta$ consecutive of APs and reinsert them at a different position, and the value of $\delta$ is set to 1, 2, or 3 in this work.
\end{itemize}

The three neighborhood search operators work within the same route. The goal of these operators is to increase the total amount of data collected while minimizing the total amount of data overflowed by adjusting the order of visiting APs
within a single route. This adjustment fully considers the temporal-spatial correlations of the TSDCMP. In each iteration, the three neighborhoods are carried out in the given sequence of APs. For each neighborhood, the best improvement strategy is adopted to exploit the best improved neighborhood found so far and then continue searching for a better solution.
\subsection{Modified DP Algorithm for Data Collection Duration Scheduling Problem}
When a UAV's trajectory is generated employing the ILS, the data collection duration for the given UAV trajectory must be determined by solving the data collection duration scheduling problem (DCDSP). Considering the key characteristic of the DCDSP, we propose a novel variation of the modified dynamic programming (MDP) algorithm to effectively tackle the problem. The MDP extends the DP by incorporating the temporal and spatial correlations of APs into the scheduling process for UAV data collection.

The basic DP approach is initially proposed by Bellman. The core idea of DP is to solve a complex problem using a divide-and-conquer strategy \cite{cormen2022introduction}. Specifically, the DP approach decomposes the complex original problem into overlapping subproblems and subsequently solves them recursively. This recursive methodology avoids redundant recalculations of solutions, enhancing computational efficiency.

We adopt the MDP algorithm to address the DCDSP. For any $k \in \{1,\ldots,N\}$ and $j \in \{0,\ldots,T\}$, let $f(k,j)$ denote the maximum total amount of data collected by selecting the first $i$ APs, with a total available service time equal to $j$ and benefit $p_k$. That is, let
\begin{equation} \label{dynamicprog}
\begin{small}
\begin{aligned}
   f(k,j)  \\=&
   \begin{cases}
   \max\{f(k-1,j),f(k-1,j-T_i)+p_k\} &\text{if } j \geq T_i,\\
   f(k-1,j), &\text{otherwise}.
   \end{cases}
\end{aligned}
\end{small}
\end{equation}
This is the MDP recursive function for the DCDSP, which enables us to compute $f(k,j)$ using Algorithm \ref{MDP-recursive-function}.

The algorithm \ref{MDP-recursive-function} outlines the general framework of the MDP algorithm for solving the DCDSP. Initially, the parameters for both the DCDSP and the MDP algorithm are input~(lines 1-3). Subsequently, the MDP algorithm is then utilized to compute the amount of data collected by the UAV, the data overflow for each AP, and the flight time between the APs (lines 4-24). Finally, the flight time $T^{f}$ is updated, and the optimal data collection during for the UAV at each AP is determined based on the state transition path (lines 27-30).
\begin{algorithm}[htb]
\begin{small}
\caption{MDP Algorithm for the DCDSP}\label{MDP-recursive-function}
\KwIn{The sequence of APs $V=\{v_0, v_1, \ldots, v_0\}$ in a given route, the total remaining service time $T$, the maximum service time per AP $\check{T}$, and $\epsilon$ is a relatively small positive real number.}
\KwOut{The optimal data collection duration for each AP $\check{T}^*$}
Initialize the total time spent in flight $T^{f} \leftarrow 0$, the data collection duration for each AP $T_i^c \in \check{T}$, $Trace \leftarrow \emptyset$, $DP[0:N,0:T]\leftarrow \epsilon$, $Data[0:N,0:N] \leftarrow \emptyset$\;
Initialize $Data[0:N,0:N] \leftarrow 0$ to record the amount of data collected and overflow for each state\;
Initialize the benefits of each state $p$\;
\For{$i\leftarrow 1$ to $N$}{
    Update the flight time $T^{f}$\;
    \While{$j \leq i*\check{T}$ and $j \leq T$}{
        \For{$k\leftarrow 1$ to $j$}{
            $t \leftarrow T^f + j - k$\;
            Calculate the amount of data stored $D_i$ by node $i$ in the state\;
            \If{$D_i - \alpha_i*k > D_{th_i}$}{
                \Continue
            }
            \If{$DP[i-1][j-k] == \epsilon$}{
                \Continue
            }
            $p \leftarrow value(Data[i-1][j-k], t, v_i)$\;
            \If{$p > DP[i][j]$}{
                $DP[i][j] \leftarrow p$\;
                $Trace[i][j] \leftarrow p$\;
                Update $Data[i][j]$\;
            }
        }
    }
}
Select the data collection duration $T$ corresponding to the maximum benefits in DP\;
\For{$i\leftarrow N$ to $1$}{
    Update the flight time $T^{f}$\;
    $\check{T}^*[i] \leftarrow Trace[i][T]$\;
    $T \leftarrow T \setminus Trace[i][T]$\;
}
\end{small}
\end{algorithm}
\section{Simulation Studies}\label{Simulation-Studies}
In this section, we present simulation experiments to assess the effectiveness of the developed model and the performance of the proposed solution approach for the TSDCMP.
\vspace{-0.1in}
\subsection{Simulation Settings}
The simulations are conducted on a workstation equipped with an AMD Ryzen Threadripper 3970X 32-Core CPU running at 3.79GHz, 192GB RAM, and operating on Ubuntu 20.04.4 LTS using a single thread. All simulation programs were developed using C++ programming language. It is important to note that all variables in this work are assigned uniform units for simplicity.
\subsubsection{Instance Generation}
Given the absence of an established benchmark for the TSDCMP, we adapt a subset of Solomon's instances to create the test dataset for this study\footnote{The Solomon instances are obtained from http://neo.lcc.uma.es/vrp/vrp-instances/.}, wherein we proportionally magnify the coordinates in original instances by ten times. In addition, to tailor the dataset for the TSDCMP, we included pertinent parameters pertaining to both the UAV and each AP, such as the UAV's battery energy and the AP's storage capacity. The dataset consists of instances with cluster-located APs (C), random-located APs (R), and random-cluster-located APs (RC). The number of APs belongs to $\{15,20,30,40\}$. Therefore, there are a total of 12 instances. Each instance is named as ``Type $|N|$", for example, ``C15" refers to the instances with fifteen cluster-located APs. The details of these instances can be obtained from \url{https://sites.google.com/view/dbjiao/homepage}.
\subsubsection{Compared Algorithms}
Due to the lack of existing algorithms for solving the TSDCMP, we compare the proposed algorithm \emph{ILS-MDP} with following two modified high-performance algorithms, namely \emph{ILS-MCS} \cite{yu2022team} and \emph{CHPBILM} \cite{JIAO2024108084}. These serve as benchmark algorithms within the same framework, differing only in allocation of arrival time and data collection duration. This enables their application to our problem, thereby verifying the effectiveness of the proposed model and algorithm in this work.
\begin{itemize}
  \item \emph{ILS-MCS} is a hybrid heuristic algorithm that integrates a modified coordinate search (MCS) into an iterated local search. Specifically, the iterated local search is adopted to construct and improve the routes, while the MCS is employed to determine the service time at each AP.
  \item \emph{CHPBILM} is a cooperative memetic algorithm that integrates that an efficient heuristic algorithm, variant population-based incremental learning (PBIL), and a maximum-residual-energy-based local search (MRELS) strategy. Specifically, the efficient heuristic algorithm is employed to improve the flight trajectory of the UAV, while the PBIL with MRELS is utilized to determine the number of restored unit circles for each degraded area on the grasslands.
\end{itemize}

It is worth noting that MCS optimizes continuous time variables in the original problem, while time is treated as a discrete variable in our scenario. To adapt MCS to our problem, we round down the results obtained by MCS to obtain a solution. In addition, we represent integer solutions directly in binary form in PBIL. We determine the optimal value for each binary bit through a search process, and ultimately convert the binary value to decimal to obtain the final solution.
\subsubsection{Performance Metrics}
We adopt the following four performance metrics to evaluate the effectiveness of the proposed model and algorithm.
\begin{itemize}
  \item \emph{\textbf{Objective function value}} is to maximize the total amount of data collected while minimizing the total amount of data overflowed for all APs by the UAV during the data collection tour, as defined by Eq. \eqref{maximizing-data-collection} in Section \ref{Problem-Formation}.
  \item \emph{\textbf{Data collection efficiency}} $\eta$ is the ratio of the total amount of data collected by the UAV, as defined by Eq. \eqref{total-amount-data-collection}, to the sum of both the total amount of data collected by the UAV and the amount of data overflow for all APs, as defined by Eq. \eqref{total-data-overflow}, during the data collection tour in Section \ref{Problem-Formation}. This metric is defined as follows:
      \begin{equation} \label{data-collection-efficiency}
      \begin{small}
      \eta = \frac{C}{C + O}.
      \end{small}
      \end{equation}
  \item \emph{\textbf{Ratio of data collection time}} $\mu$  is defined as the proportion of the total data collection time to the sum of the total flight time of the UAV and the data collection time. This metric is defined as follows:
      \begin{equation} \label{Ration-data-collection-time}
      \mu = \frac{T^c}{T^c + T^f}.
      \end{equation}
  \item \emph{\textbf{Algorithm runtime}} refers to the running time of each algorithm for each instance.
\end{itemize}
\subsubsection{Parameter Settings}
In the system model described in the Section \ref{System-Model}, the parameters of the UAV are set according to the values employed in \cite{zeng2019energy}: $P_o=79.85$, $U_{tip}=120$, $\mathrm{v}_0=4.03$, $d_0=0.6$, $\rho=1.225$, $s=0.05$, and $A=0.503$. Additionally, the UAV speed $\mathrm{v}$ is set to 10. The battery capacity carried by the UAV is set to 37000, 46000, 132000, and 141000 depending on different types of networks. The penalty parameter $\eta$ is set to 15 to determine the best-performing solutions, the data collection thresholds varies from 0.75 to 0.90 in increasing steps of 0.05 based on the size of the networks, and the signal-to-noise ratio range from 2.4 to 6.4. In addition, the other communication-related parameters are set as follows, consistent with \cite{samir2019uav}: $B = 1000$, $\beta = -50$, $\sigma^2 = -110$, $\tau = 1$, and $\tilde{H} = 1$. The transmit power $p$ of AP for uploading data link varies depending on different instances, with values ranging from $p \in (5.52,14.36)$.

We conducted preliminary experiments to fine-tune the hyper-parameter for all algorithms to determine the best-performing configurations. In ILS-MDP, the maximum number generation $T^{DMP}_{max}$ is set to $15$.  Moreover, in the comparison algorithms, the parameters of the ILS-MCS are set based on the values used in \cite{yu2022team}: $\alpha=10$, $\alpha_{min}=1$, and $\|\vec{e_i}\|=0.01$. In ILS-PBIL, the population size and maximum iteration number of PBIL are respectively set to $P=30$ and $T^{PBIL}_{max}=300$. The learning rate is set to  $\alpha=0.4$ and the proportion of elite solution $\theta$ is set to $33\%$. The PBIL algorithm utilizes binary encoding with 5 bits, indicating that the maximum data collection time for each AP is 31. The maximum execution time of all algorithms is limited to 1800 seconds.

In the simulations, each algorithm runs independently ten times in each scenario\footnote{In this work, the instance and scenario can be used interchangeably.} to minimize the influence of randomness on the test results as much as possible. The average values of these ten runs are taken as the simulation results, providing a more reliable assessment of the algorithm's performance.
\subsection{Simulation Results}
We first evaluate the performance of different algorithms by varying the number of APs from 15 to 40 across three different network types in Table \ref{Comparison-of-the-objective}. For solving the TSDCMP using ILS-MDP, ILS-MCS, and ILS-PBIL separately, we present the best, worst, and average objective function values; the standard deviation among ten runs; and the average total computing time in seconds for ten runs. Boldface is used to highlight the best value, as well as the shortest computing time among the different algorithms.
\begin{table*}[htbp]
\begin{small}
\renewcommand\arraystretch{1.2}
\centering
\begin{threeparttable}
\caption{Comparison of the objective function values among the three different algorithms for 12 instances.}
\label{Comparison-of-the-objective}
\begin{center}
\setlength{\tabcolsep}{0.10em}
\begin{tabular}{ccccccccccccccccccc}
\hline\hline%
\multirow{2}{*}{Instance}  &\multicolumn{1}{c}{}  &\multicolumn{5}{c}{ILS-MCS}   &\multicolumn{1}{c}{} &\multicolumn{5}{c}{ILS-PBIL}   &\multicolumn{1}{c}{} &\multicolumn{5}{c}{ILS-MDP}\\
\cline{3-7} \cline{9-13} \cline{15-19}
    &&Best&Worst&Avg. &Std.&Time~(s)   &&Best&Worst&Avg. &Std.&Time~(s)  &&Best&Worst&Avg. &Std.&Time~(s)\\
\hline
C15&&195667&195667&195667 &\textbf{0.00}&0.92&& 165701&143921&155556 &6091.06&1148.56&& \textbf{212973}&\textbf{212973}&\textbf{212973}&\textbf{0.00}&\textbf{0.76}\\
C20&&222053&222053&222053 &\textbf{0.00}&10.64&& 213080&186701&201146 &8216.21&2437.15&& \textbf{234621}&\textbf{234621}&\textbf{234621}&\textbf{0.00}&\textbf{3.01}\\
C30&&263381&263381&263381 &\textbf{0.00}&140.29&& 171183&159411&165999 &4304.71&5249.83&& \textbf{301694}&\textbf{301694}&\textbf{301694}&\textbf{0.00}&\textbf{27.66}\\
C40&&457536&457536&457536 &\textbf{0.00}&917.67&& 58404&58404&58404 &\textbf{0.00}&6184.15&& \textbf{523513}&\textbf{523513}&\textbf{523513}&\textbf{0.00}&\textbf{134.11}\\
R15&&7437&7437&7437 &\textbf{0.00}&4.82&& \textbf{27922}&5883&\textbf{12959} &7637.82&1887.31&& 10824&\textbf{10824}&10824&\textbf{0.00}&\textbf{0.95}\\
R20&&128470&128470&128470 &\textbf{0.00}&10.28&& 120169&96385&111683 &7144.84&2657.99&& \textbf{134937}&\textbf{134937}&\textbf{134937}&\textbf{0.00}&\textbf{3.09}\\
R30&&174107&174107&174107 &\textbf{0.00}&222.04&& -1123740&-1126140&-1124940 &120000&3022.29&& \textbf{219612}&\textbf{182080}&\textbf{193696}&13075.21&\textbf{49.27}\\
R40&&8636&-92528&-1474 &30352.93&936.94&& -1838772&-1838772&-1838772 &\textbf{0.00}&6522.08&& \textbf{199404}&\textbf{199404}&\textbf{199404}&\textbf{0.00}&\textbf{177.55}\\
RC15&&171462&171462&171462 &\textbf{0.00}&2.21&& 176951&162727&167699 &3639.13&1236.33&& \textbf{181299}&\textbf{175524}&\textbf{180723}&1735.50&\textbf{0.76}\\
RC20&&157095&157095&157095 &\textbf{0.00}&13.86&& 145139&128080&134852 &4709.21&1919.86&& \textbf{161469}&\textbf{161469}&\textbf{161469}&\textbf{0.00}&\textbf{2.09}\\
RC30&&92709&92709&92709 &\textbf{0.00}&126.13&& -244860&-258871&-251879 &7004.50&2018.17&& \textbf{153827}&\textbf{153827}&\textbf{153827}&\textbf{0.00}&\textbf{34.27}\\
RC40&&15871&15871&15871 &\textbf{0.00}&390.59&& -151144&-151144&-151144 &\textbf{0.00}&7025.09&& \textbf{445480}&\textbf{445480}&\textbf{445480}&\textbf{0.00}&\textbf{126.43}\\
\hline\hline
\end{tabular}
\raggedright
\begin{tablenotes}
\item[1] Standard deviation is abbreviated as ``Std.", and average value is abbreviated as ``Avg."
\item[2] The algorithm running time is the average of 10 runs, rounded to two decimal places.
\item[3] All average values are rounded.
\end{tablenotes}
\vspace{5pt}
\end{center}
\end{threeparttable}
\end{small}
\end{table*}
\begin{figure*}[htb]
    \centering
    \subfigure[ILS-MCS]{
        \includegraphics[width=2.1in]{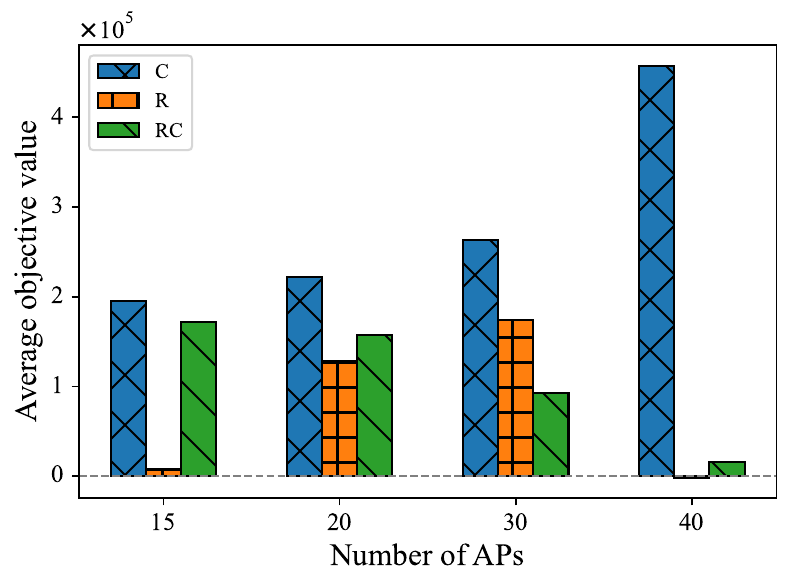}
    }
    \subfigure[ILS-PBIL]{
        \includegraphics[width=2.2in]{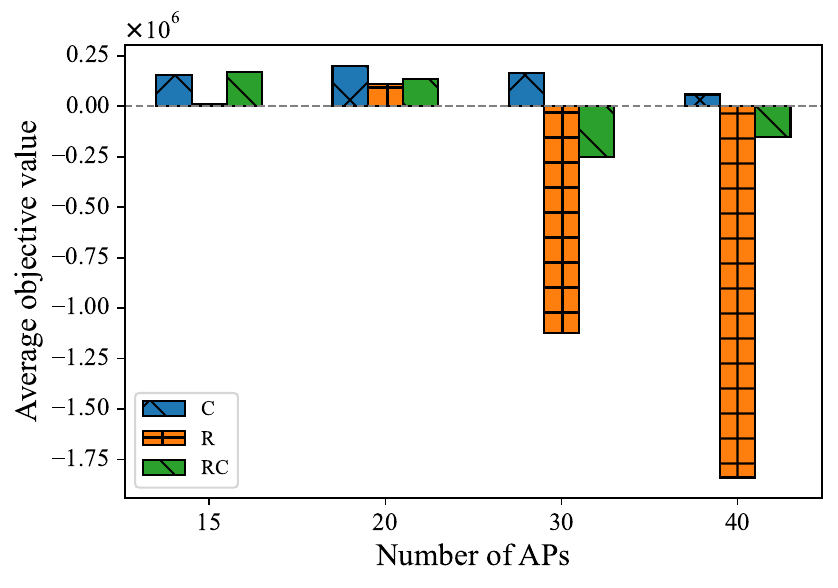}
    }
    \subfigure[ILS-MDP]{
        \includegraphics[width=2.1in]{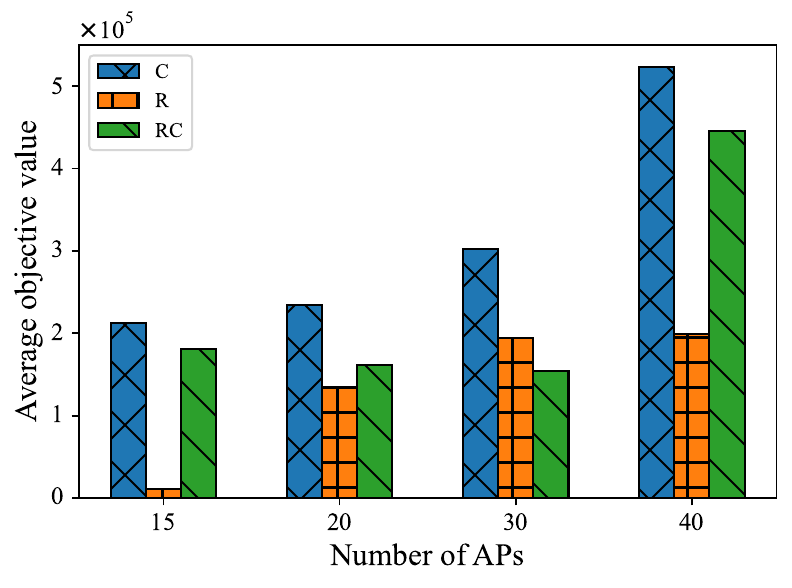}
    }
\caption{Comparison of the average objective function value of different algorithms by varying the number of APs from 15 to 40 across three different network types.} \label{Objective-value-Network-Type}
\end{figure*}

Table \ref{Comparison-of-the-objective} illustrates that the objective function value increases with the number of APs for each of the three algorithms across different network types. Meanwhile, from Fig.~\ref{Objective-value-Network-Type}, we observe that in networks with the same number of APs, the average objective function value of C-type networks is typically much larger than that of other types, while the average objective function value of R-type and CR-type networks depends on the distribution of AP locations. The reason behind this phenomenon is that the UAV in the solution provided by the different algorithms is more inclined to collect data from clusters with many APs in C-type networks, considering the temporal-spatial correlations during APs' data collection. Consequently, more energy can be dedicated to data collection, resulting in an increase in the objective function value. Conversely, in the other two types of networks, the relatively scattered distribution of AP locations leads to increased energy consumption for UAV flying between different APs. In some cases, APs may even experience data overflow due to delayed data collection, resulting in objective function values that are not as large as those observed in C-type networks. 

As the UAV's energy increases proportionally, the presence of more APs and higher network density leads to increased data collection by the UAV. Consequently, data overflow decreases, resulting in a larger objective function value. Moreover, as observed in Table \ref{Comparison-of-the-objective}, the proposed algorithm ILS-MDP demonstrates significantly superior performance compared with the other two algorithms in terms of the best value, worst value, and average value, except for instance R15. In addition, from the results shown in Table \ref{Comparison-of-the-objective}, it can be seen that the ILS-MDP could generally find better solutions than the other two algorithms within considerably less computing time and also demonstrates good stability. Furthermore, from Table \ref{Comparison-of-the-objective} and Fig.~\ref{Objective-value-Network-Type}, we observed that when employing algorithms ILS-MCS and ILS-PBIL to tackle the large-scale R30 and R40 instances, the objective function values appeared to be negative. This occurred because within the maximal computing time, these two algorithms, especially ILS-PBIL, converge slowly and  prompt premature termination of the iteration process, failing to produce high-quality solutions.
\begin{table*}[htbp]
\renewcommand\arraystretch{1.2}
\centering
\begin{threeparttable}
\center
\caption{Comparison of the UAV data collection efficiency among the three different algorithms for 12 instances.}
\label{UAV-data-collection-efficiency}
\setlength{\tabcolsep}{0.95em}
\begin{tabular}{ccccccccccccc}
\hline\hline
\multirow{2}{*}{Instance}&\multicolumn{1}{c}{}&\multicolumn{3}{c}{ILS-MCS}&\multicolumn{1}{c}{}&\multicolumn{3}{c}{ILS-PBIL}&\multicolumn{1}{c}{}&\multicolumn{3}{c}{ILS-MDP}\\
\cline{3-5}\cline{7-9}\cline{11-13}
&&$C$&$O$&$\eta$&&$C$&$O$&$\eta$&&$C$&$O$&$\eta$\\
\hline
C15&&238987&2888&\textbf{0.988}&&209715&2934&0.986&&261528&3237&0.987\\
C20&& 222053 & 0 & \textbf{1.000}  && 213080  & 0 & \textbf{1.000} && 236046  & 95 & 0.999\\
C30&& 353951 & 6038 & \textbf{0.983} && 319728 & 9903 & 0.970 && 418499 & 7787 & \textbf{0.981} \\
C40&& 585591  & 8537 & 0.986 && 387924 & 21968 & 0.946 && 627448 & 6929 & \textbf{0.989}\\
R15&& 203697 & 13084 & 0.940&& 212872 & 12330 & \textbf{0.945}  && 210984 & 13344 & 0.941 \\
R20&& 260635  & 8811 & 0.967 && 268309  & 9876 & 0.964 && 251262  & 7755 & \textbf{0.970} \\
R30 && 394397 & 14686 & 0.964 && 281806 & 93863 & 0.750 && 336297 & 7779 & \textbf{0.977}\\
R40&& 501761 & 32875 & 0.939  && 311058 & 143322 & 0.684 && 451692 & 16819 & \textbf{0.964} \\
RC15&& 192072 & 1374 & 0.992 && 198221 & 1418 & 0.992 && 200079 & 1252 & \textbf{0.993} \\
RC20&& 197535  & 2696 & 0.987 && 182984  & 2523 & 0.986 && 197889  & 2428 & \textbf{0.989}\\
RC30&& 409809 & 21140 & 0.951 && 390825 & 42379 & 0.902 && 364292 & 14031 & \textbf{0.963}\\
RC40&& 582406  & 37769 & 0.939 && 587666  & 49254 & 0.923 && 585220 & 9316 & \textbf{0.984}\\
\hline\hline
\end{tabular}
\begin{tablenotes}
\item[1] The data collection efficiency $\eta$ is the average of 10 runs, rounded to three decimal places.
\end{tablenotes}
\vspace{5pt}
\end{threeparttable}
\end{table*}

We then assess the algorithms' performance by calculating the data collection efficiency of UAV. Table \ref{UAV-data-collection-efficiency} showcases the efficiency of UAV data collection for each instance. It is evident that the data collection efficiency of UAV, when using the proposed ILS-MDP algorithm, exceeds that of the other two algorithms in over $75\%$ of the instances. This implies that the proposed ILS-MDP algorithm can effectively utilize the temporal-spatial correlations among APs for data collection. Consequently, the UAV's flight energy consumption is reduced, which in turn leaves more energy available for collect data.

Moreover, we also investigate the algorithms' performance by examining the ratio of data collection time. As observed in Table \ref{Ratio-of-effective-data-collection-time}, the proposed ILS-MDP algorithm shows a slight advantage over ILS-MCS algorithm in terms of ratio of data collection time, despite a slight performance gap in some instances. However, when compared with the ILS-PBIL algorithm, the advantage is notably significant. Moreover, when the UAV achieves maximum the amount of data collection during a data collection process, specifically when the best solution is reached, the best UAV path determined by each algorithm for each instance is depicted in Fig. \ref{Visualization-Network-Type}.
\begin{figure*}[htbp]
    \centering
    \subfigure[C15]{
        \includegraphics[width=2.1in]{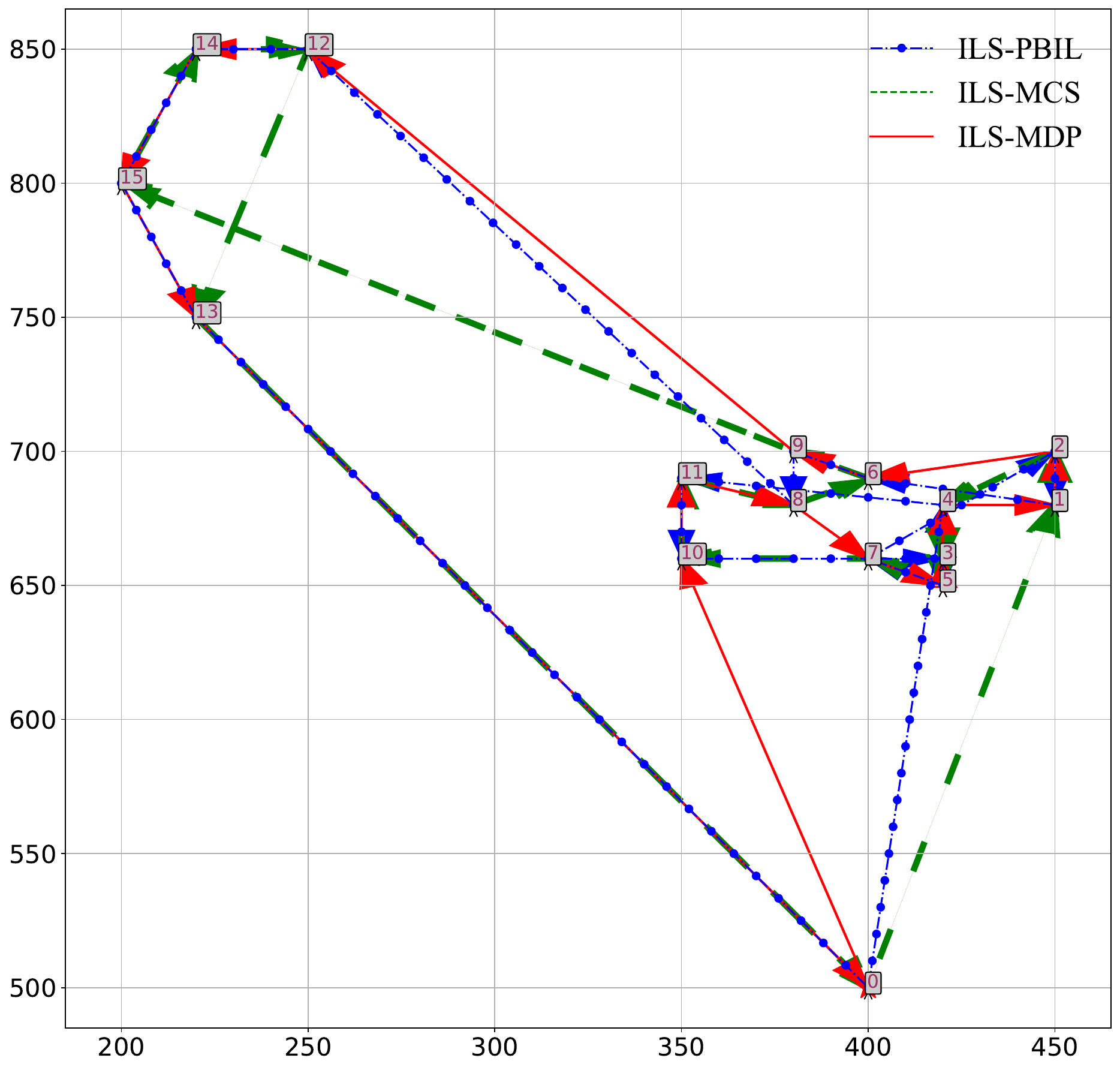}
    }
    \subfigure[R15]{
        \includegraphics[width=2.1in]{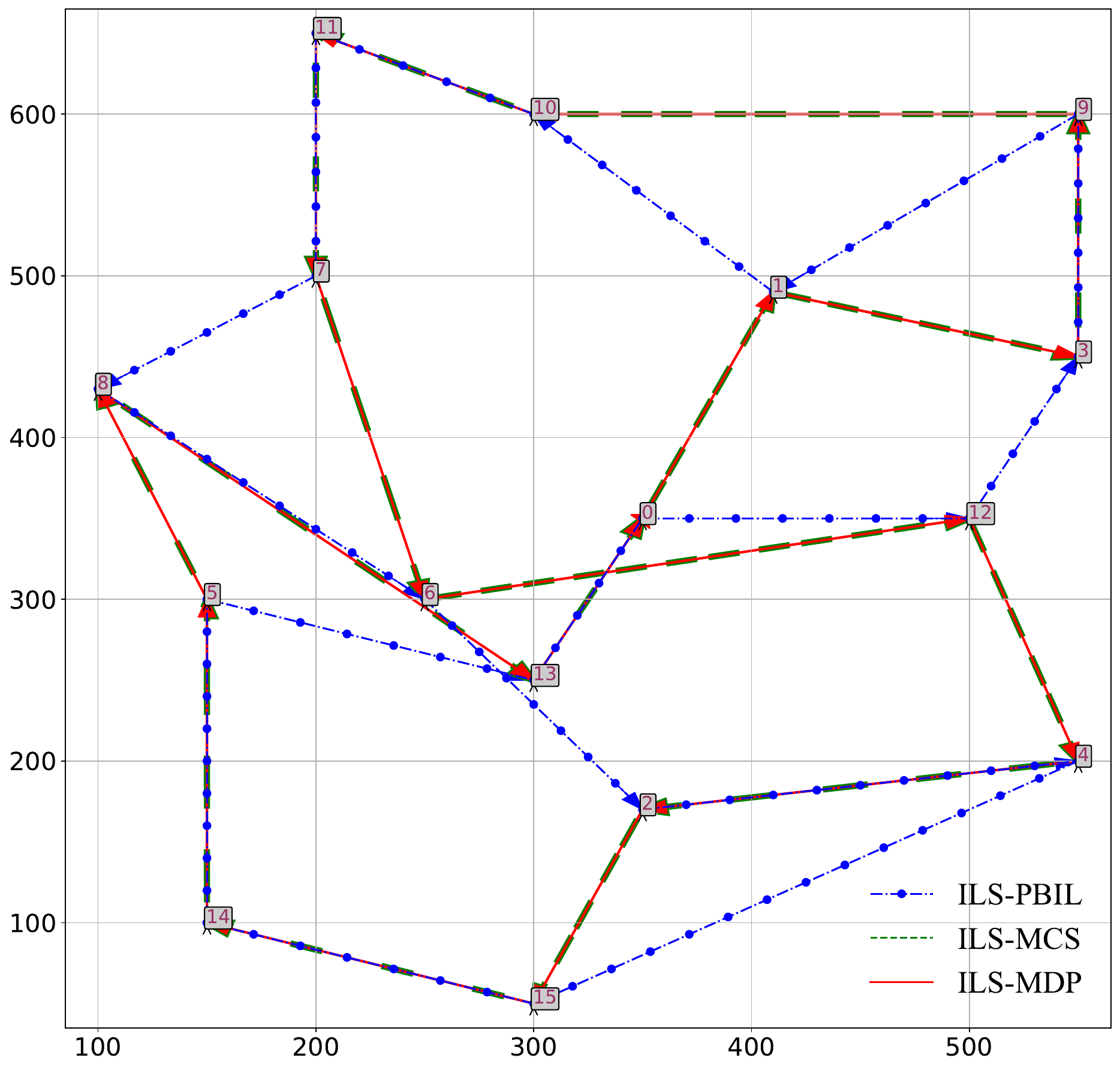}
    }
    \subfigure[RC15]{
        \includegraphics[width=2.1in]{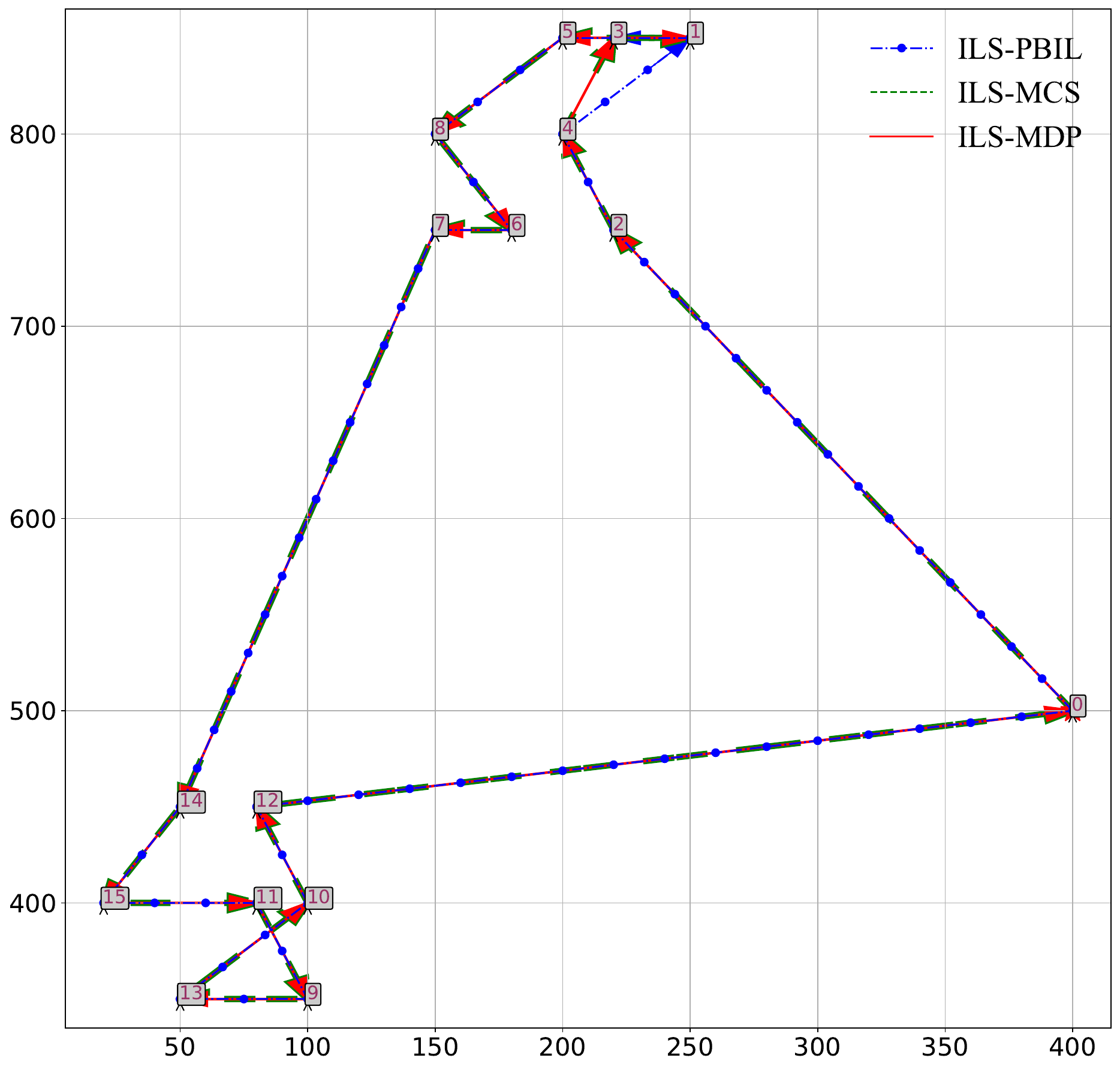}
    }
    \subfigure[C20]{
        \includegraphics[width=2.1in]{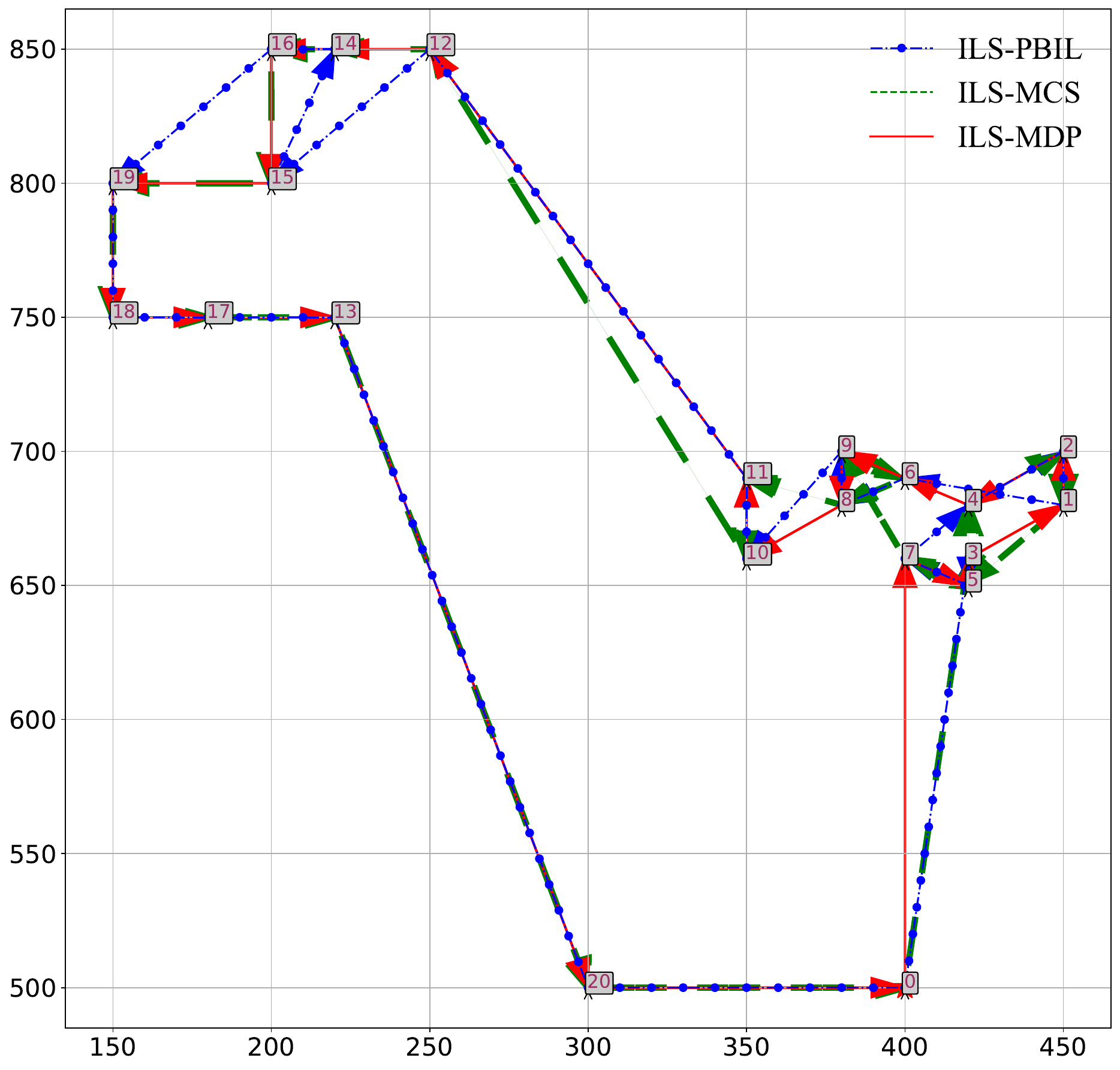}
    }
    \subfigure[R20]{
        \includegraphics[width=2.1in]{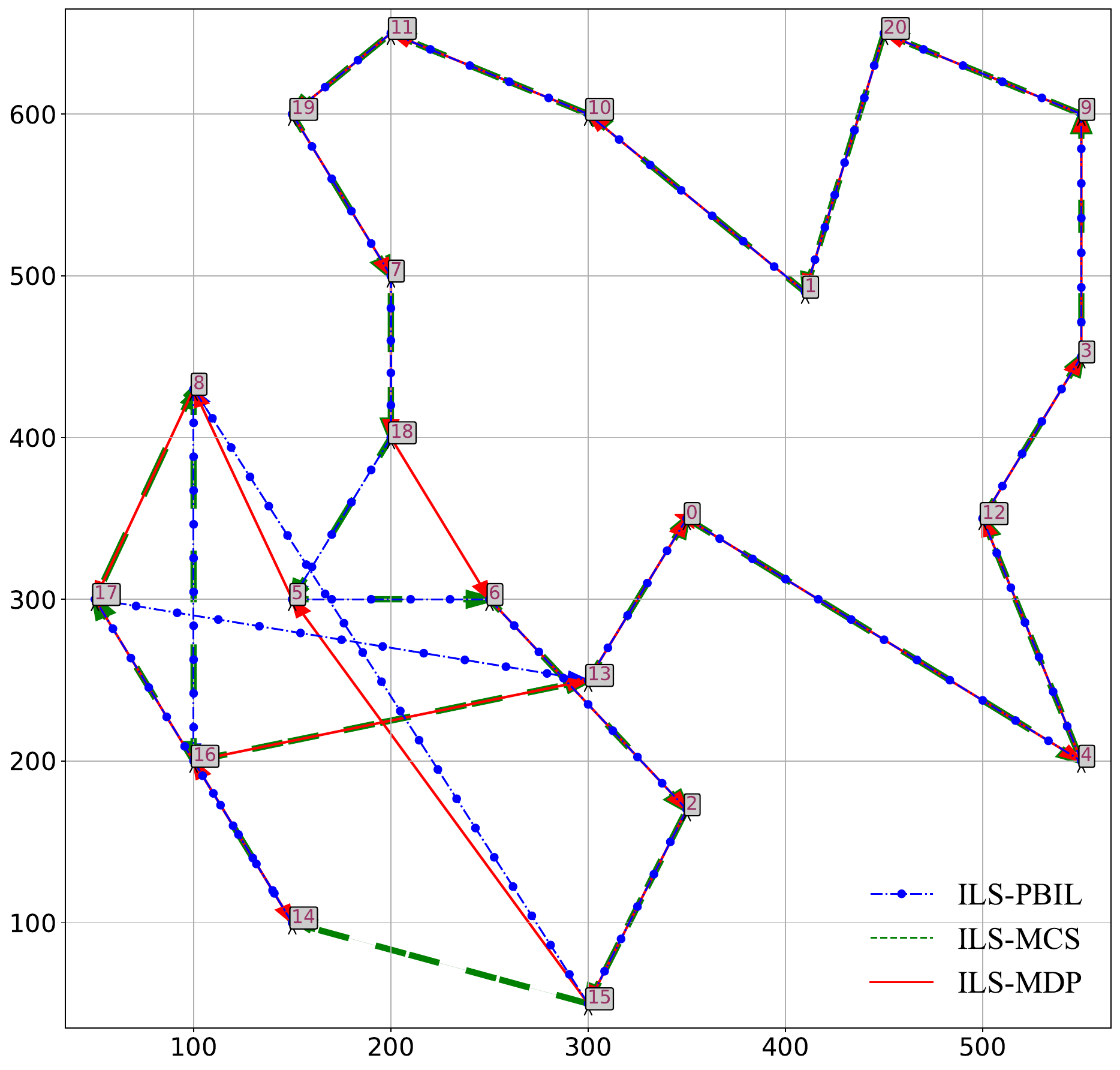}
    }
    \subfigure[RC20]{
        \includegraphics[width=2.1in]{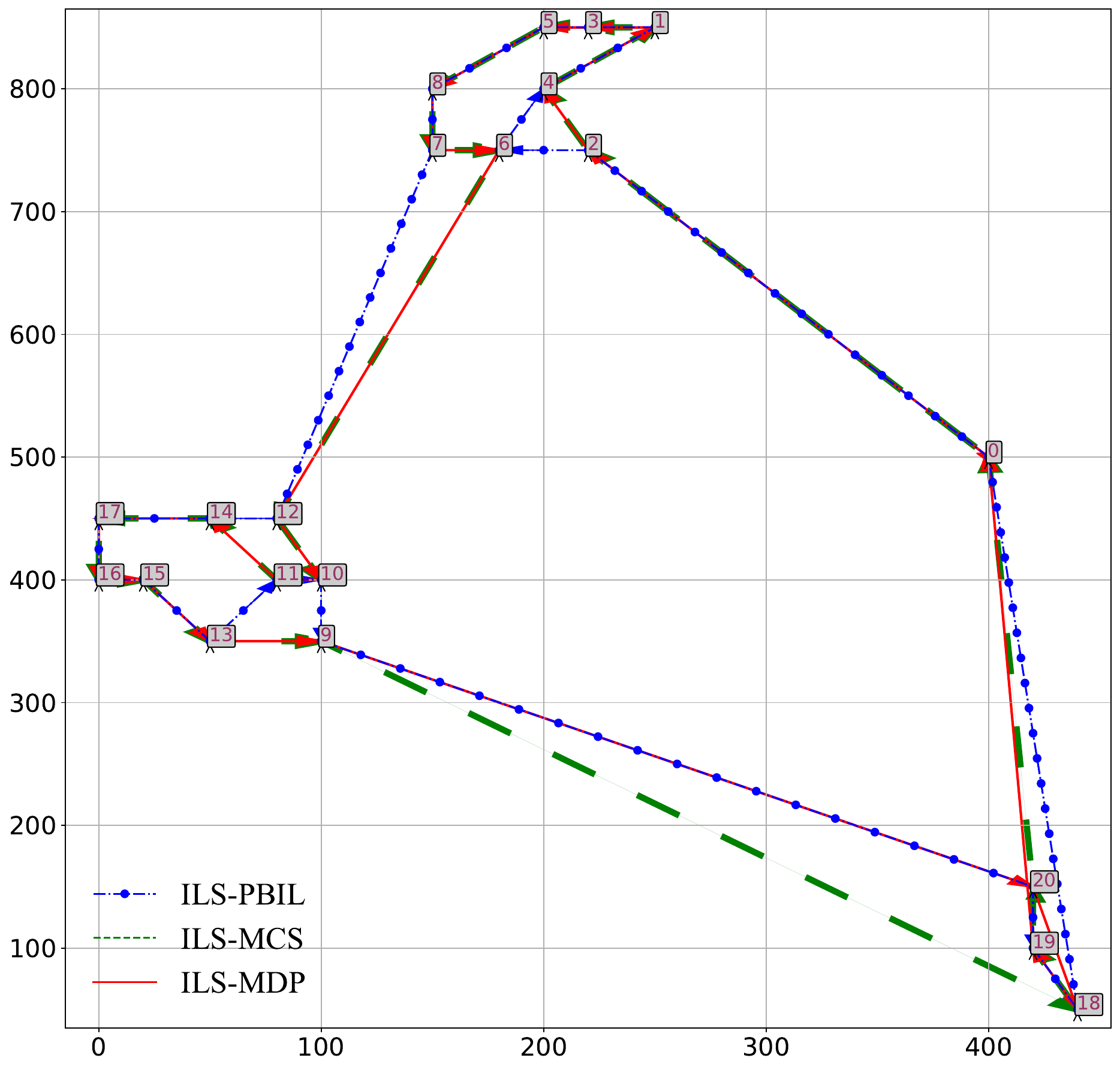}
    }
    \subfigure[C30]{
        \includegraphics[width=2.1in]{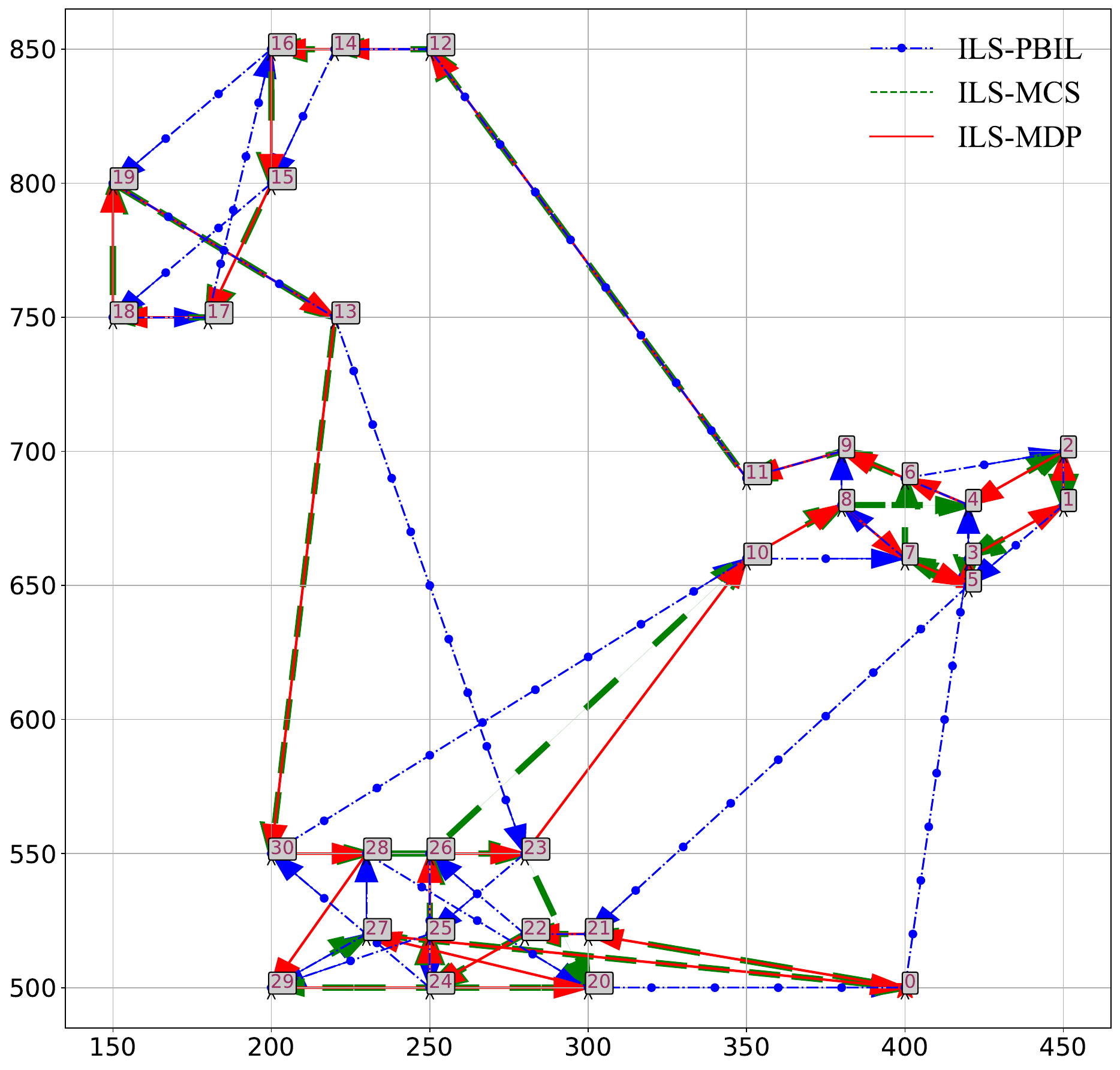}
    }
    \subfigure[R30]{
        \includegraphics[width=2.1in]{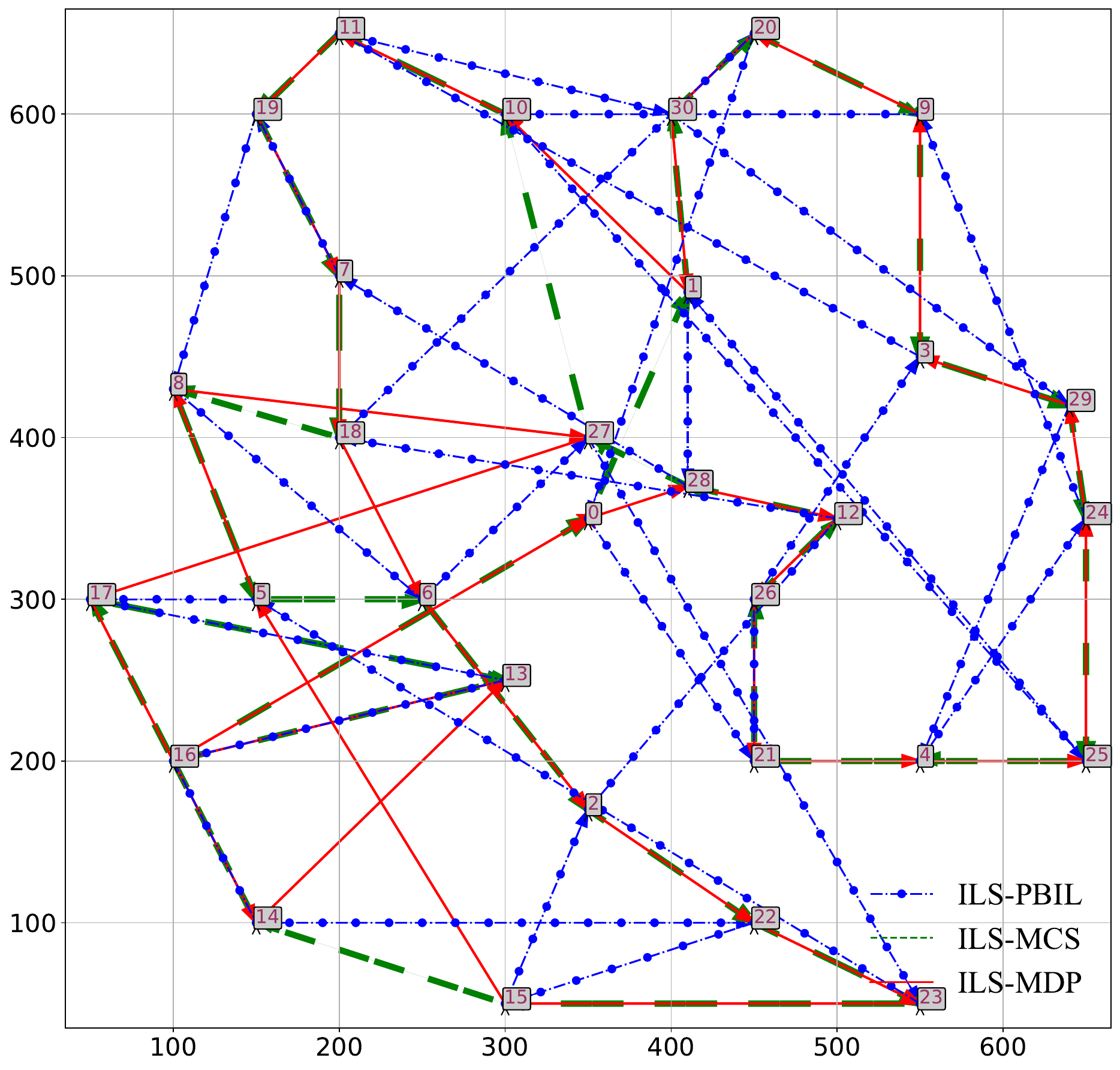}
    }
    \subfigure[RC30]{
        \includegraphics[width=2.1in]{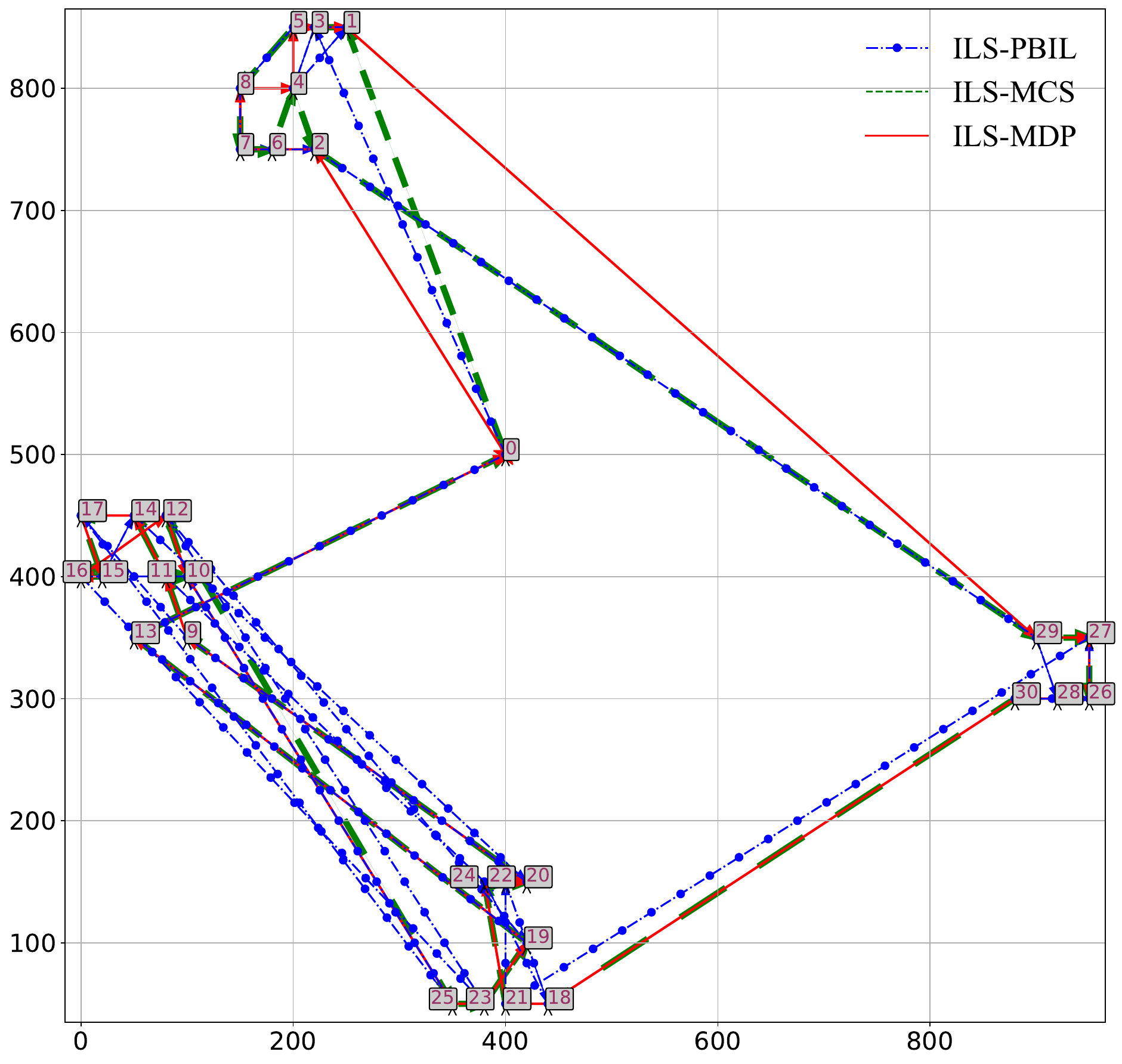}
    }
    \subfigure[C40]{
        \includegraphics[width=2.1in]{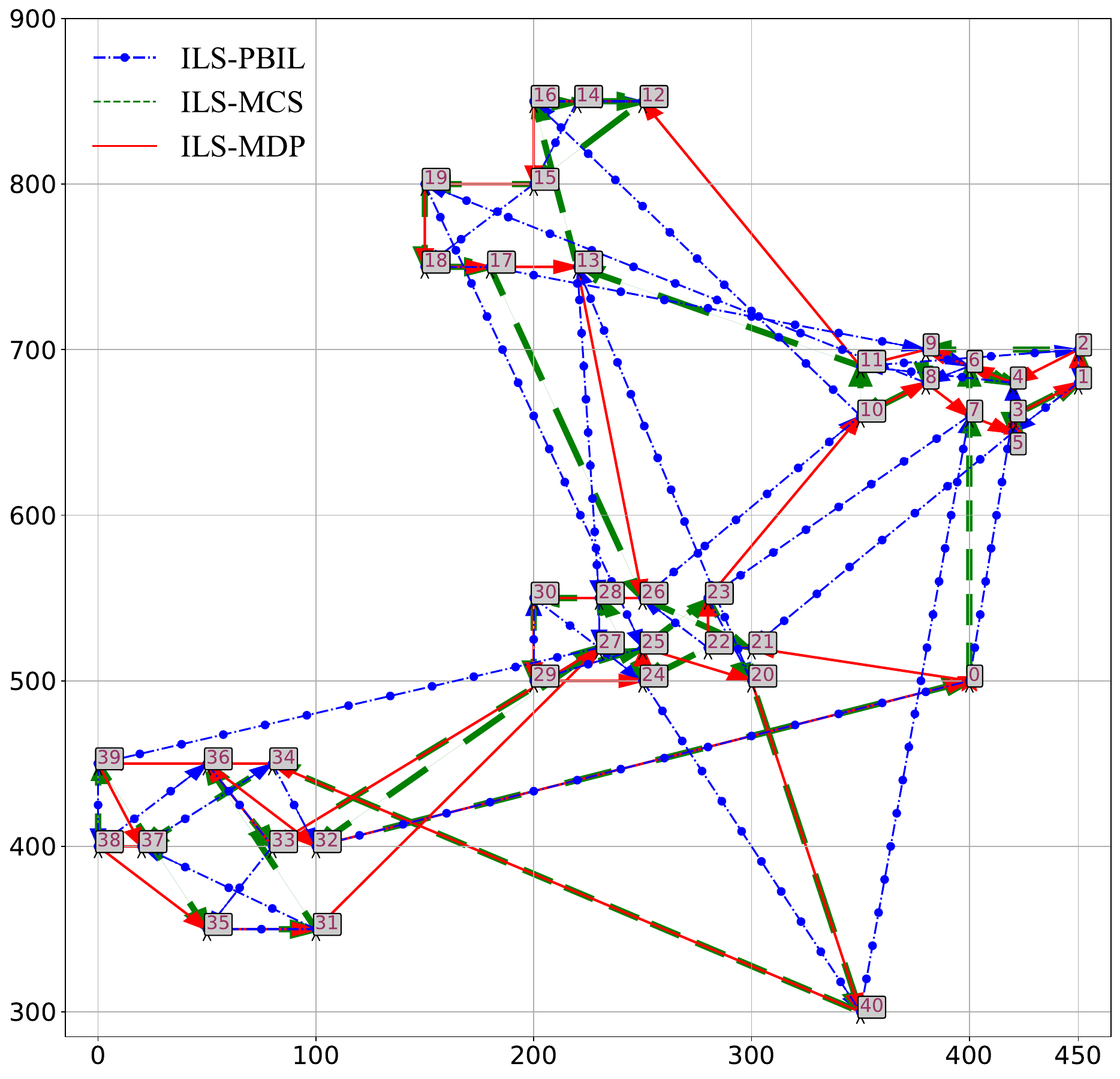}
    }
    \subfigure[R40]{
        \includegraphics[width=2.1in]{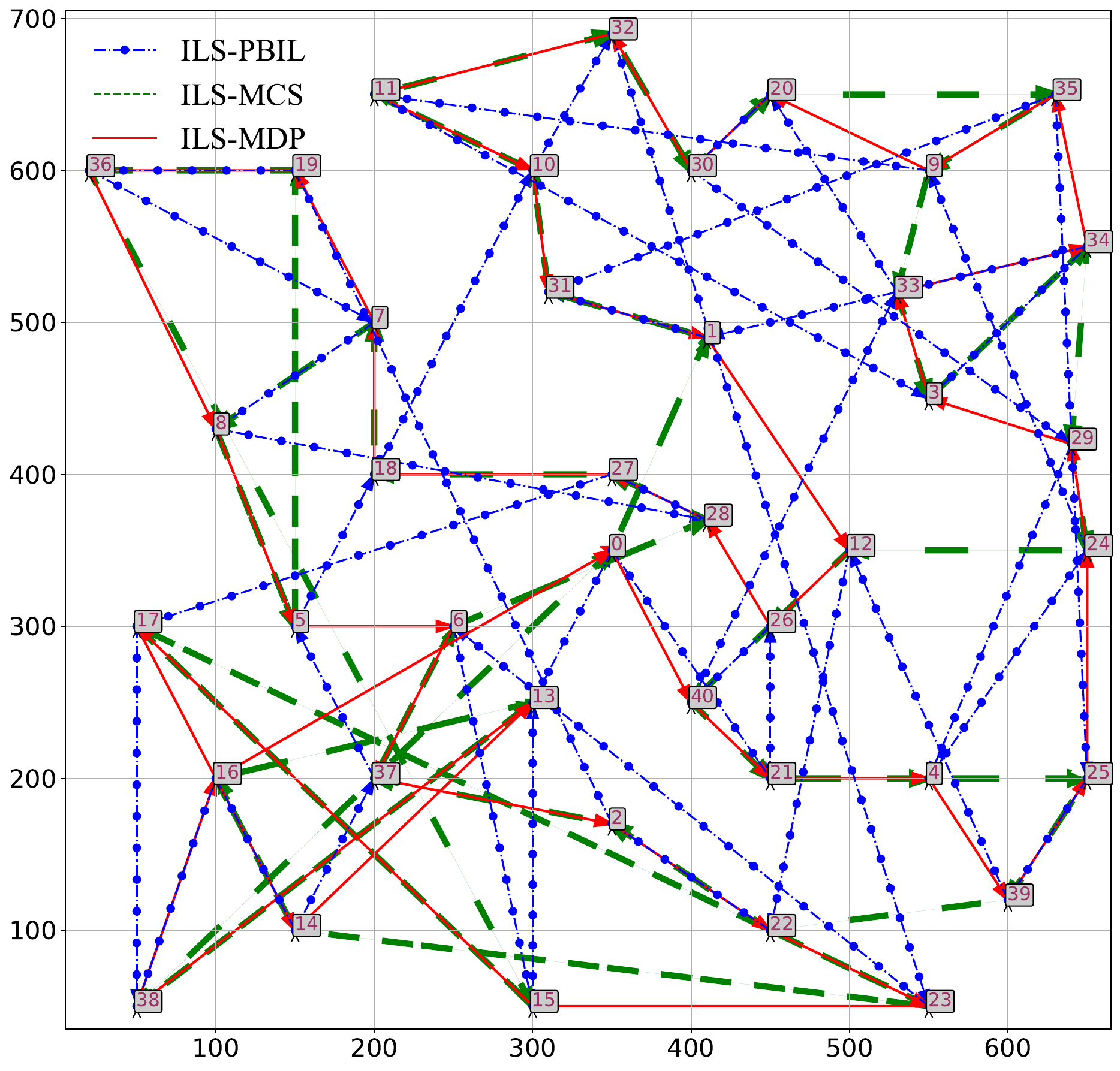}
    }
    \subfigure[RC40]{
        \includegraphics[width=2.1in]{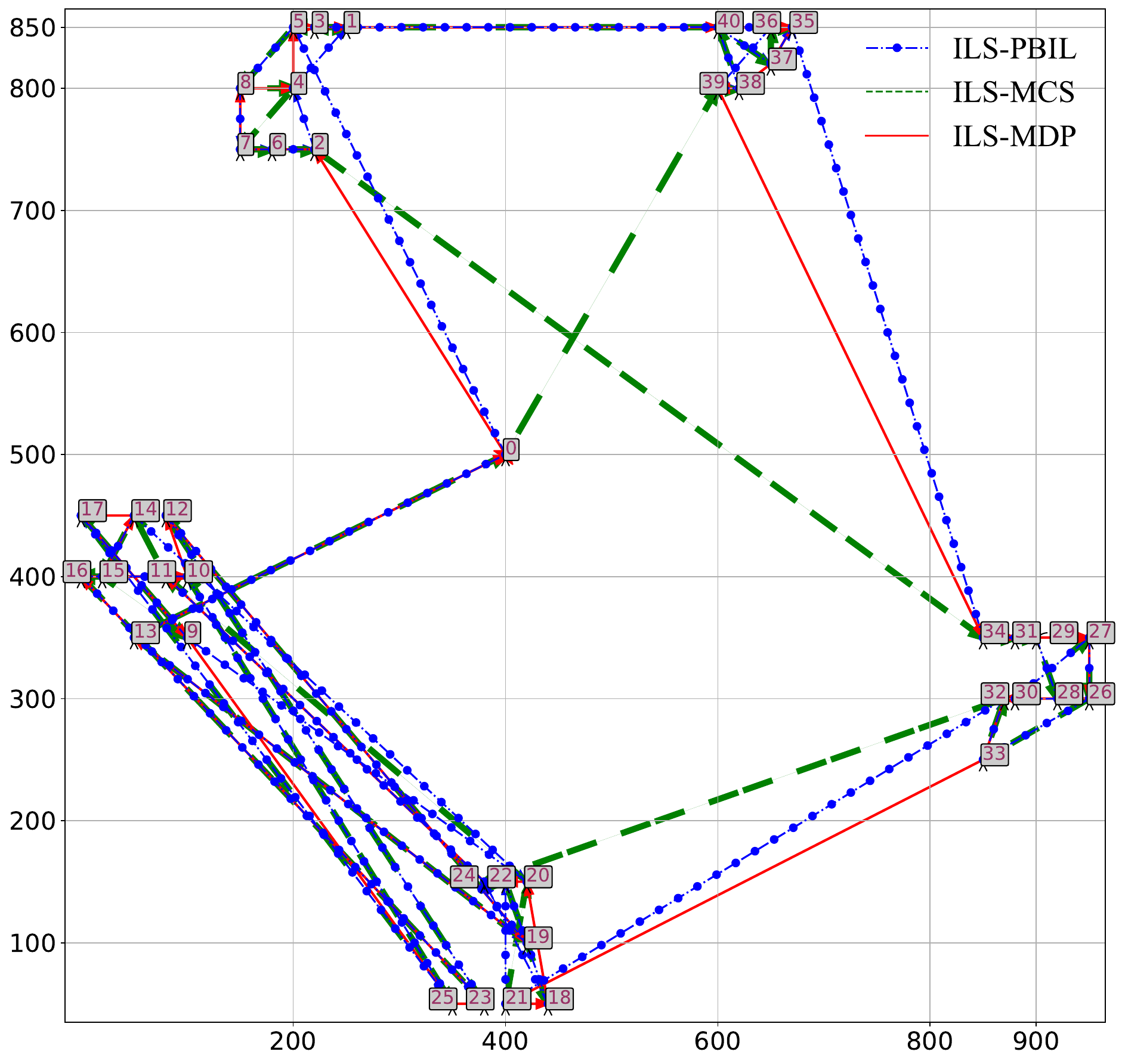}
    }
\caption{Visualization of the best paths generated by three different algorithms in three different network scenarios for 12 instances.} \label{Visualization-Network-Type}
\end{figure*}
\begin{table*}[htbp]
\renewcommand\arraystretch{1.2}
\centering
\begin{threeparttable}
\center
\caption{Ratio of effective data collection time among the three different algorithms for 12 instances.}
\label{Ratio-of-effective-data-collection-time}
\small
\setlength{\tabcolsep}{0.75em}
\begin{tabular}{ccccccccccccc}
\hline\hline
\multirow{2}{*}{Instance}&\multicolumn{1}{c}{}&\multicolumn{3}{c}{ILS-MCS}&\multicolumn{1}{c}{}&\multicolumn{3}{c}{ILS-PBIL}&\multicolumn{1}{c}{}&\multicolumn{3}{c}{ILS-MDP}\\
\cline{3-5}\cline{7-9}\cline{11-13}
&&Flight&Service&$\mu$&&Flight&Service&$\mu$&&Flight&Service&$\mu$\\
\hline
C15&&115 & 85 &0.425&& 122 & 73 &0.374&& 107 & 91&\textbf{0.459} \\
C20&&130 & 80 &0.381&& 133 & 78 &0.370&& 123 & 85&\textbf{0.410}\\
C30&& 179 & 143 &0.444&& 218 & 129 &0.372&& 177 & 170&\textbf{0.490}\\
C40&& 325 & 258 &0.443&& 462 & 170 & 0.269&&307 & 277&\textbf{0.474}\\
R15&& 278 & 110 &0.284 &&261 & 113 &\textbf{0.302}& &278 & 112&0.287\\
R20&& 307 & 152 &\textbf{0.331} &&336 & 157 & 0.318&&312 & 142&0.312\\
R30 && 416 & 216 &\textbf{0.342}&& 755 & 153 &0.169&& 452 & 183&0.288\\
R40&& 692 & 263 &0.276 &&907 & 158 &0.148 &&552 & 234&\textbf{0.298}\\
RC15&& 164 &  96 & \textbf{0.369}&&162 & 97 & 0.374&&164 &  96&\textbf{0.369}\\
RC20&& 221 & 106 &\textbf{0.324} &&222 & 97 &0.304 &&224 & 105&0.319\\
RC30&& 450 & 214 &\textbf{0.322} &&760 & 197 &0.206& &436& 186&0.300\\
RC40&& 747 & 256 & 0.255&&814 & 247 & 0.232&&566 & 254&\textbf{0.310}\\
\hline\hline
\end{tabular}
\begin{tablenotes}
\item[1] The ratio of data collection time $\mu$ is the average of 10 runs, rounded to three decimal places.
\end{tablenotes}
\vspace{5pt}
\end{threeparttable}
\end{table*}

Finally, to emphasize the advantages of our proposed algorithm and model, and enhance the credibility of the experimental results, we perform a statistical test. All comparisons are validated by using the nonparametric Wilcoxon rank-sum test with significance level of $p=0.05$. The statistical outcomes are detailed in Table \ref{Statistical-testing}. From Table \ref{Statistical-testing}, it is evident that the $p-$values are all less than 0.05 except for the instances of R15 with the ILS-PBIL algorithm. These statistical results suggest that the observed differences are statistically significant.
\begin{table*}[htbp]
\renewcommand\arraystretch{1.2}
\centering
\begin{threeparttable}
\caption{Statistical testing of the performance among the three different algorithms for 12 instances.}
\label{Statistical-testing}
\begin{center}
\setlength{\tabcolsep}{0.30em}
\begin{tabular}{cccccccccccccc}
\hline\hline%
\multirow{2}{*}{Algorithm pairs}  &\multicolumn{1}{c}{}  &\multicolumn{12}{c}{$p$-value}\\
\cline{3-14}
&&C15&C20&C30&C40&R15&R20&R30&40&RC15&RC20&RC30&RC40\\
\hline
ILS-MDP $vs.$ ILS-PBIL&&0.0000&0.0000&0.0000&0.0000&0.443($l$)&0.0000&0.0001&0.0000&0.0001&0.0000&0.0000&0.0000\\
ILS-MDP $vs.$ ILS-MCS&&0.0000&0.0000&0.0000&0.0000&0.0000&0.0000&0.0000&0.0000&0.0000&0.0000&0.0000&0.0000\\
\hline\hline
\end{tabular}
\begin{tablenotes}
\item[1] The Wilcoxon rank-sum test is conducted with significance level $p$ is $0.05$, and $l$ indicates that the statistical result is nonsignificant.
\end{tablenotes}
\vspace{5pt}
\end{center}
\end{threeparttable}
\end{table*}

\section{Conclusion} \label{Conclusion}
This work addressed a variant of TOP with time-varying profits for the TSDCMP within GMENs during a data collection tour. It represents the first research effort to solve this problem by jointly optimizing the flight trajectory, arrival time, and data collection duration for the UAV at each AP. The goal was to determine how the UAV manages its limited energy, coupled routing and time scheduling, and the state of APs and UAV arrival time to minimize the volume of overflowed data while maximizing the volume of data collected for APs. We formulated the TSDCMP for UAV-enabled GMENs as a  mixed-integer programming model in a single trip under the realistic constraints. To address this problem, we developed a novel CHTSC algorithm, which integrates a modified MDP into an iterated local search. This algorithm fully considers the temporal and spatial relationships between consecutive service requests from APs. The simulation results indicated that the mixed-integer programming model accurately represents the characteristics of the TSDCMP within UAV-enabled GMENs. Furthermore, the proposed CHTSC algorithm outperformed two superior algorithms across twelve different scale instances.

Future research directions include exploring collaborative data collection schemes using multiple UAVs to overcome the limitations of a single UAV's data collection capabilities. In addition, expanding the study to incorporate autonomous decision-making framework that accounts for variable UAV altitudes in complex grassland terrains is essential. What is more, the model and solution framework developed here can be extended to other domains, such as emergency rescue operations.


%

%

\section*{Acknowledgment}
This work was supported in part by National Natural Science Foundation of China~(Grant No. 62233003, 62176109, U21A20183, and 62272210), the Fundamental Research Funds for the Central Universities (Grant No. lzujbky-2021-2 and CHD300102223109), and the Natural Science Foundation of Gansu Province, China~(Grant No. 20JR10RA640).

The authors would like to thank Mr. Z. Guo and Dr. L. Yuan for kind help and valuable discussions. We also thank the editors and anonymous referees for their insightful comments and helpful suggestions which significantly improve the manuscript's quality.

\ifCLASSOPTIONcaptionsoff
  \newpage
\fi



%
%
%

\bibliographystyle{IEEEtran}

\bibliography{Reference-UAV}

%
%
\begin{IEEEbiography}[{\includegraphics[width=1in,height=1.25in,clip,keepaspectratio]{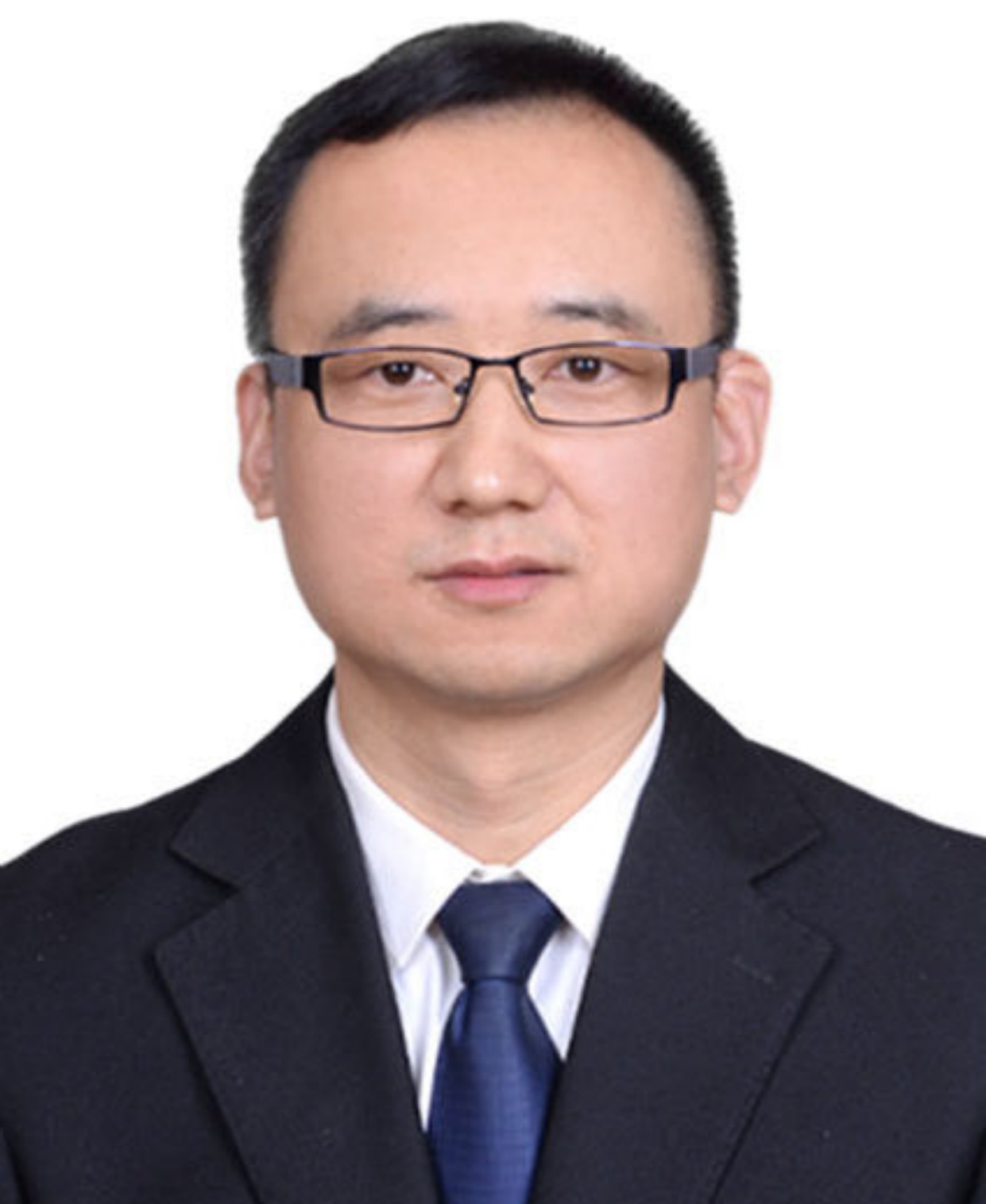}}]{Dongbin Jiao}
(Member, IEEE) received the B.S. degree in Information and Computing Sciences from Lanzhou University of Technology, Lanzhou, China, in 2008, the M.Eng. degree in Software Engineering from Dalian Maritime University, Dalian, China, in 2012, and the Ph.D. degree in Control Science and Engineering at Xi'an Jiaotong University, Xi'an, China, in 2019.  

Since March 2020, he has been a Lecturer with the School of Information Science and Engineering, Lanzhou University, China. His research interests include network resource optimization, UAV networks, intelligent unmanned autonomous systems, combinatorial optimization, evolutionary learning, and their application in ecology and environment protection. He is currently a reviewer for more than ten international journals, including TCOM, TTE, IoT-J, and serves as a TPC member for several international conferences, such as IEEE ICC and ICCC.
\end{IEEEbiography}
\begin{IEEEbiography}[{\includegraphics[width=1in,height=1.25in,clip,keepaspectratio]{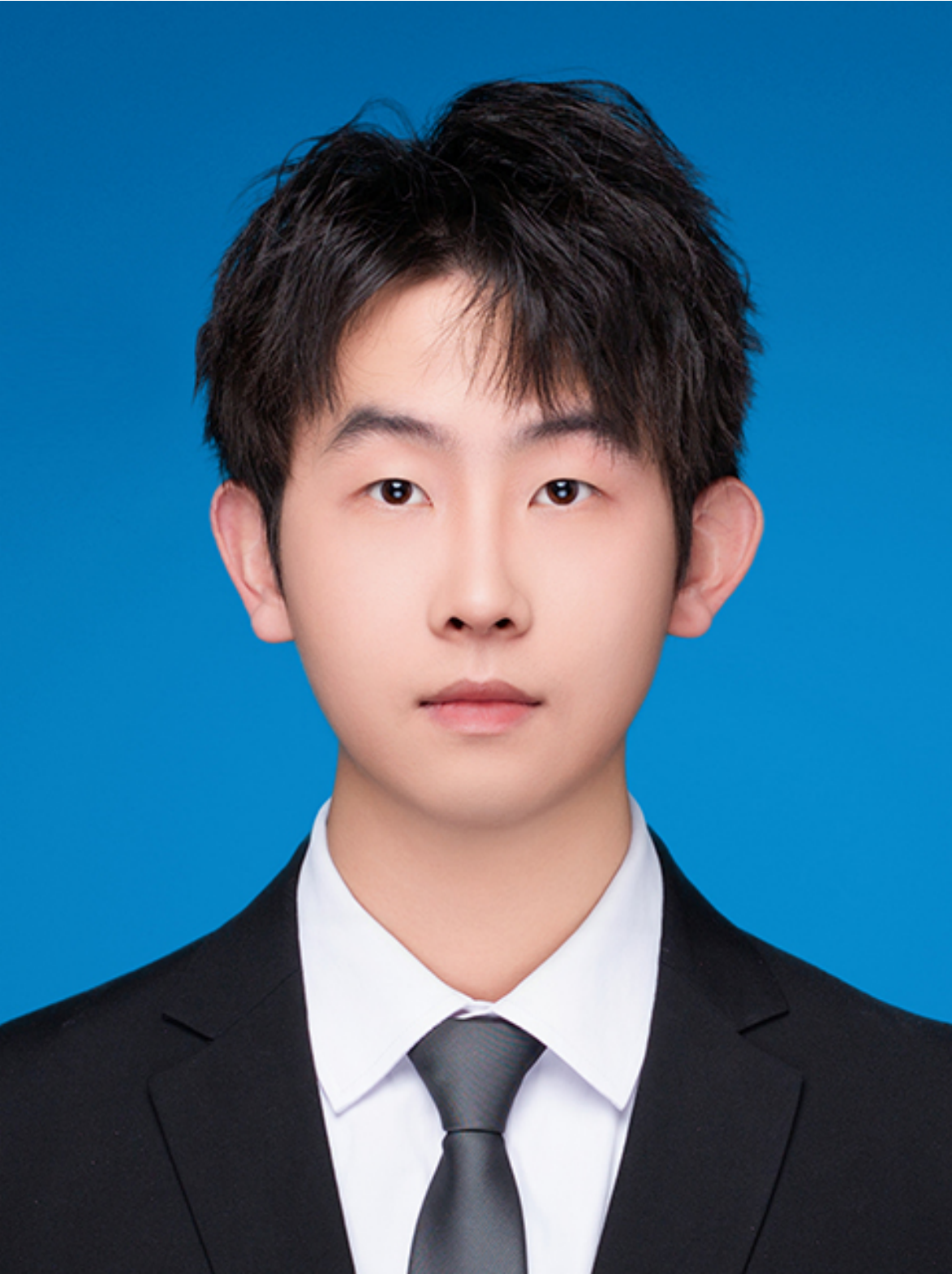}}]{Zihao Wang}
received the B.Eng. degree in Computer Science and Technology from the Lanzhou University, China, in 2024.

He is currently an engineer at Geoming AI. His research interests include combinatorial optimization and large language model.
\end{IEEEbiography}
\begin{IEEEbiography}[{\includegraphics[width=1in,height=1.25in,clip,keepaspectratio]{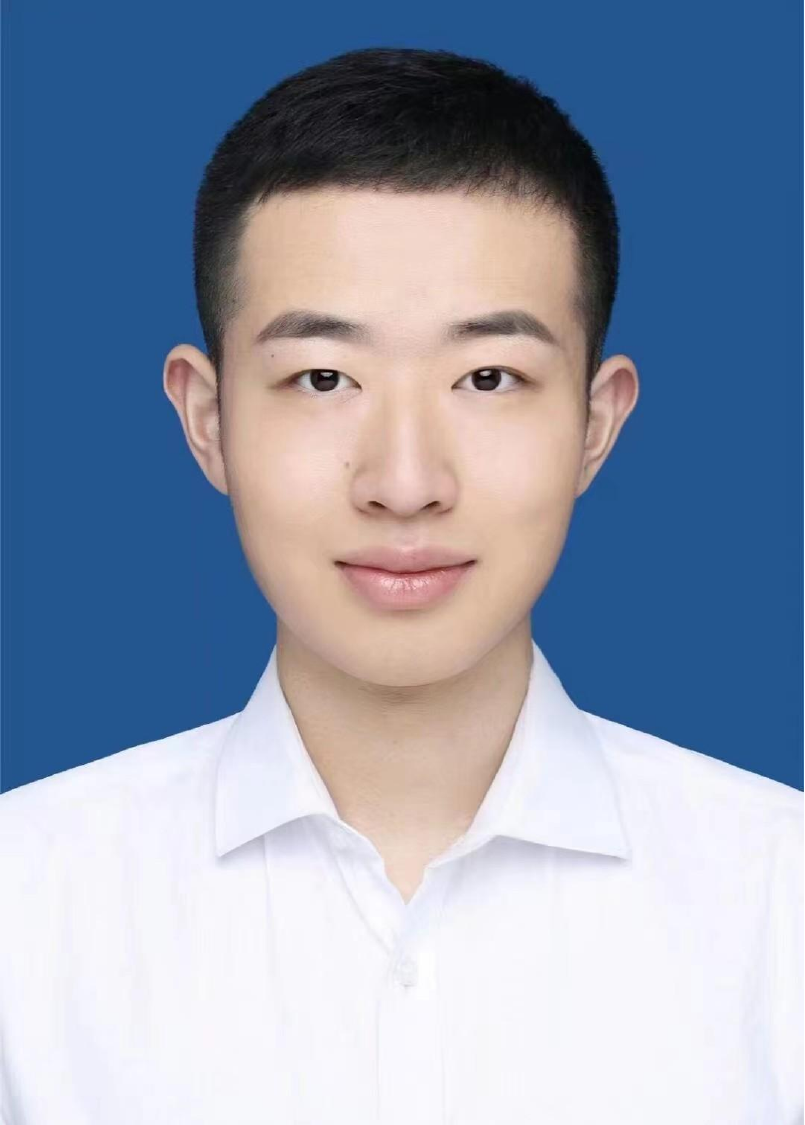}}]{Wen Fan}
received his B.Eng. degree in Computer Science and Technology from Sichuan University, Chengdu, China, in 2019 and the M.Eng. degree in the same major from Lanzhou University, Lanzhou, China, in 2024.

He is currently an engineer at Chery Automobile Co., Ltd. His research interests include IoT, UAV path planning, and reinforcement learning.
\end{IEEEbiography}
\begin{IEEEbiography}[{\includegraphics[width=1in,height=1.25in,clip,keepaspectratio]{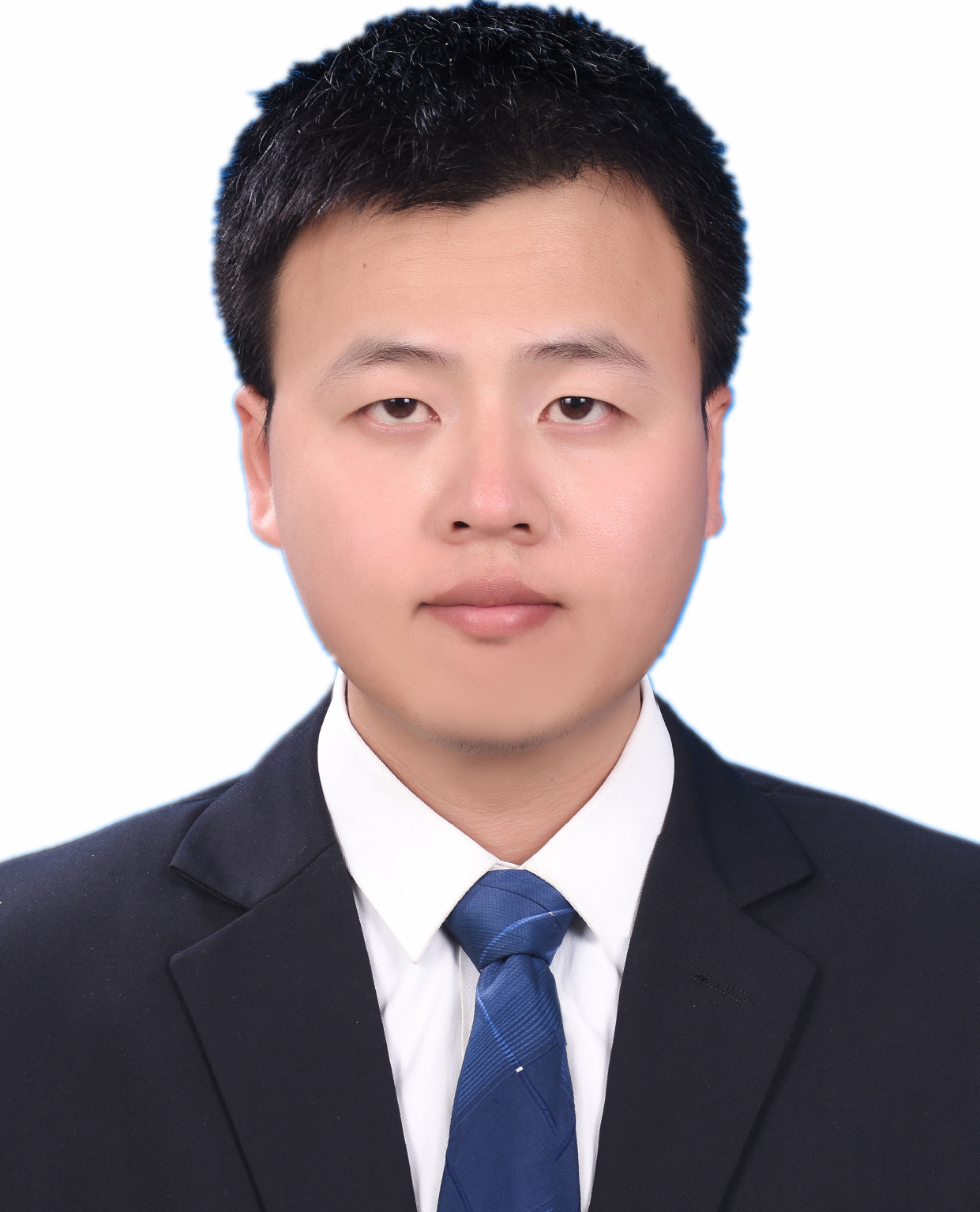}}]{Weibo Yang}
received the B.Eng. degree in Automation and the Ph.D. degree in Control Science and Engineering from Xi'an Jiaotong University, Xi'an, China, in 2013 and 2020, respectively.

Since March 2022, he has been a Lecturer with the School of Automobile, Chang'an University, China. His research interests include vehicle routing problem, integer programming, and evolutionary computation.
\end{IEEEbiography}
\begin{IEEEbiography}[{\includegraphics[width=1in,height=1.25in,clip,keepaspectratio]{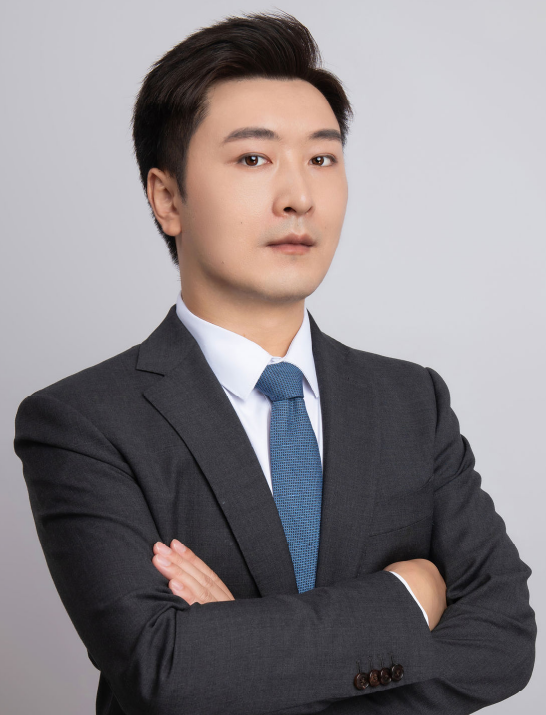}}]{Peng Yang}(Senior Member, IEEE)
is a tenure-track assistant professor jointly in the Department of Statistics and Data Science and the Department
of Computer Science and Engineering, Southern University of Science and Technology (SUSTech), China. He received his B.Sc. and Ph.D. degrees in the Department of Computer Science and Technology from University of Science and Technology of
China in 2012 and 2017, respectively. From 2017 to 2018, he was a Senior Engineer in Huawei and then he joined SUSTech. His research interests include Evolutionary Computation, Reinforcement Learning, and Multi-agent Systems.
He has published 37 papers in top journals and conferences. He has served as the reviewer for top-tier journals (TEVC, TNNLS, TIE) and the PC member of top conferences (NeurIPS, ICLR, and ICML). He is a member of IEEE Computational Intelligence Society (CIS) Evolutionary Computation Technical Committee.
\end{IEEEbiography}
\begin{IEEEbiography}[{\includegraphics[width=1in,height=1.25in,clip,keepaspectratio]{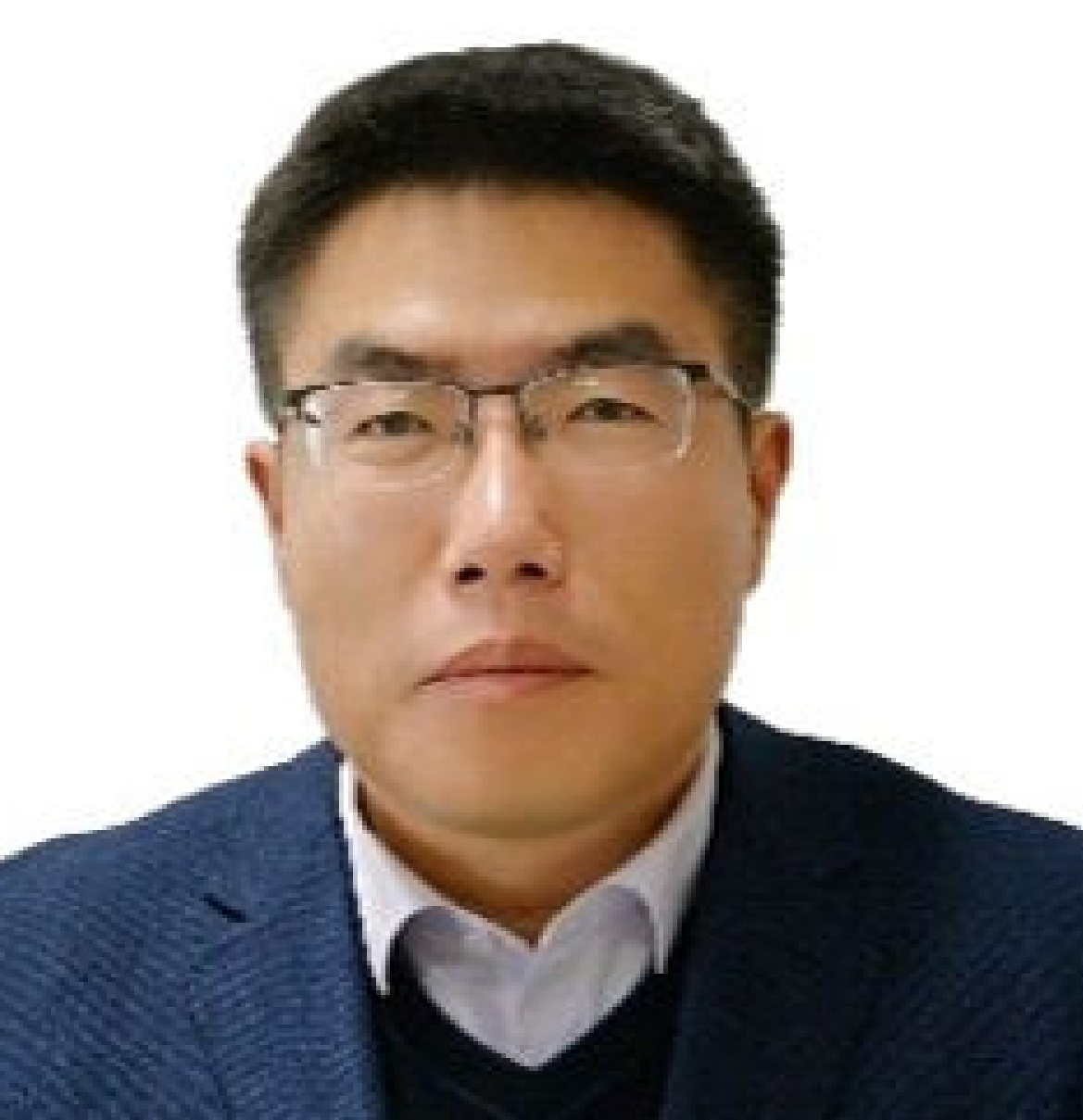}}]{Zhanhuan Shang} is a full professor at College of Ecology, Lanzhou University, China. His research interests include ecological restoration, grassland management, rural sustainable development, carbon management, and biodiversity on the Tibetan plateau. He has more than 200 publications in both international and Chinese journals and books.
\end{IEEEbiography}
\begin{IEEEbiography}[{\includegraphics[width=1in,height=1.25in,clip,keepaspectratio]{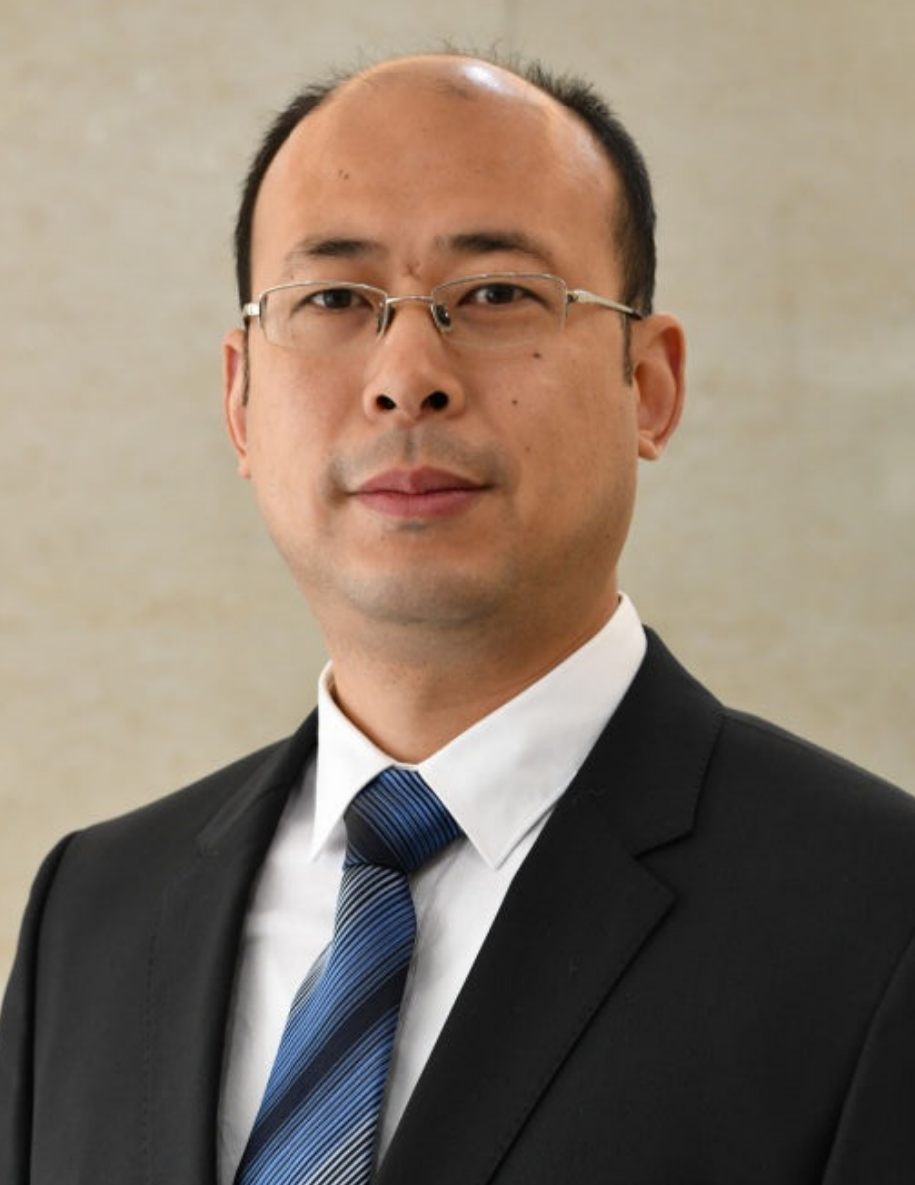}}]{Shi Yan}
(Member, IEEE) received the B.S. and M.Eng. degrees from Lanzhou University, Lanzhou, China, in 2001 and 2004, respectively, and the Dr.Eng. degree from Akita Prefectural University, Akita, Japan, in 2010.

From July 2004 to March 2007, he was an Assistant Professor with the School of Information Science and Engineering, Lanzhou University. He was a Visiting Researcher from January 2007 to March 2007 and a Visiting Research Fellow from April 2010 to September 2010 with the Department of Electronics and Information Systems, Akita Prefectural University. He was a Lecturer from April 2007 to April 2012 and an Associate Professor from May 2012 to November 2020. He is currently a Professor with the School of Information Science and Engineering, Lanzhou University. His research interests include uncertain dynamical systems, robotics, machine learning, and multidimensional system theory.
\end{IEEEbiography}

%




\end{document}